\documentclass[%
reprint,
 amsmath,amssymb,
 aps,
floatfix,
aps, nofootinbib, superscriptaddress,
]{revtex4-2}
\usepackage[dvipdfmx]{graphicx}\usepackage{dcolumn}\usepackage{bm}\usepackage[braket, qm]{qcircuit}
\usepackage[caption=false,labelformat=simple]{subfig}

\usepackage{float}
\usepackage{amsmath}
\usepackage{amsthm}
\usepackage{algorithm,algpseudocode}
\usepackage{color}
\usepackage{here}
\usepackage[title]{appendix}
\usepackage{enumitem}

\theoremstyle{definition}
\newtheorem{dfn}{Definition}[section]
\newtheorem{prop}[dfn]{Proposition}
\newtheorem{lem}[dfn]{Lemma}
\newtheorem{thm}[dfn]{Theorem}
\newtheorem{cor}[dfn]{Corollary}
\newtheorem{rem}[dfn]{Remark}

\newtheorem{conj}[dfn]{Conjecture}

\newif\ifoverleaf
\overleaftrue

\ifoverleaf
\usepackage[whole]{bxcjkjatype}
\else
\usepackage{bxcjkjatype}
\fi
\usepackage[normalem]{ulem}
\usepackage{xcolor} 
\usepackage{comment}

\newif\ifhidedone
\hidedonetrue

\usepackage{soul}
\newif\ifcomment
\commenttrue

\ifcomment
\newcommand{\shumpei}[1]{\textcolor[RGB]{254,0,144}{\textbf{Shumpei: }#1}}
\newcommand{\nakajima}[1]{ \textcolor[RGB]{144,0,254}{\textbf{Nakajima: }#1}}
\newcommand{\tran}[1]{\textcolor[RGB]{60,104,234}{\textbf{Tran: }#1}}
\else
\newcommand{\shumpei}[1]{\ignorespaces} \newcommand{\nakajima}[1]{\ignorespaces} \newcommand{\tran}[1]{\ignorespaces} \fi

\ifhidedone
\newcommand{\done}[1]{\ignorespaces}
\else
\newcommand{\done}[1]{#1}
\fi

\DeclareFontEncoding{FMS}{}{}
\DeclareFontSubstitution{FMS}{futm}{m}{n}
\DeclareFontEncoding{FMX}{}{}
\DeclareFontSubstitution{FMX}{futm}{m}{n}
\DeclareSymbolFont{fouriersymbols}{FMS}{futm}{m}{n}
\DeclareSymbolFont{fourierlargesymbols}{FMX}{futm}{m}{n}
\DeclareMathDelimiter{\VERT}{\mathord}{fouriersymbols}{152}{fourierlargesymbols}{147}

\begin{document}
\title{Coherence influx is indispensable for quantum reservoir computing}
\author{Shumpei Kobayashi}
 \affiliation{Department of Creative Informatics, The University of Tokyo, Japan}
 \author{Quoc Hoan Tran}
 \affiliation{Next Generation Artificial Intelligence Research Center (AI Center), The University of Tokyo, Japan}
\author{Kohei Nakajima}
 \affiliation{Next Generation Artificial Intelligence Research Center (AI Center), The University of Tokyo, Japan}
 \affiliation{Department of Creative Informatics, The University of Tokyo, Japan}
 \affiliation{Department of Mechano-Informatics, The University of Tokyo, Japan}
  \thanks{Present Address for Q. H. Tran: Quantum Laboratory, Fujitsu Research, Fujitsu Limited}
\date{\today}             
\begin{abstract}

Echo state property (ESP) is a fundamental property that allows an input-driven dynamical system to perform information processing tasks. Recently, extensions of ESP to potentially nonstationary systems and subsystems, that is, nonstationary ESP and subset/subspace ESP, have been proposed. In this paper, we theoretically and numerically analyze the sufficient and necessary conditions for a quantum system to satisfy nonstationary ESP and subset/subspace nonstationary ESP. Based on extensive usage of the Pauli transfer matrix (PTM) form, we find that (1) the interaction with a quantum-coherent environment, termed \textit{coherence influx}, is indispensable in realizing  nonstationary ESP, and (2) the spectral radius of PTM can characterize the fading memory property of quantum reservoir computing (QRC). Our numerical experiment, involving a system with a Hamiltonian that entails a spin-glass/many-body localization phase, reveals that the spectral radius of PTM can describe the dynamical phase transition intrinsic to such a system. To comprehensively understand the mechanisms under ESP of QRC, we propose a simplified model, multiplicative reservoir computing (mRC), which is a reservoir computing (RC) system with a one-dimensional multiplicative input. Theoretically and numerically, we show that the parameters corresponding to the spectral radius and coherence influx in mRC directly correlates with its linear memory capacity (MC). Our findings about QRC and mRC will provide a theoretical aspect of PTM and the input multiplicativity of QRC. The results will lead to a better understanding of QRC and information processing in open quantum systems.
\end{abstract}

\maketitle
\section{Introduction}

Recently, quantum machine learning (QML) \cite{Schuld_2018}, especially NISQ \cite{Preskill_2018}-capable unitary parametric quantum models, has gained much attention because of its near-term realizability. Recent research progresses on QML include differentiability \cite{Mitarai2018,Schuld_2019}, formulation as a kernel method \cite{Schuld2021}, the data-reuploading technique \cite{P_rez_Salinas_2020}, supervised classification \cite{Schuld_2018,Mitarai2018,Havlicek2019,Schuld_2019_Hilbert,goto_2021} and representation learning \cite{lloyd2020quantum}. However, the training of variational parameters in many QML models suffers from a difficulty that originates from the barren plateau (BP) problem \cite{McClean_2018}. BP causes flatness in the loss function landscape when evaluating observables of QML models, necessitating a problematic number of quantum measurements to precisely evaluate the gradient. Another research direction is to utilize dissipation and decoherence for machine learning \cite{Kubota_2023, Domingo2023, sannia2022dissipation, sannia2023engineered}, where algorithms even incorporate natural quantum dynamics as an information processing medium that inevitably includes dissipations.

Quantum reservoir computing (QRC) \cite{Fujii_2017} is a parameterless temporal QML model that utilizes potentially uncontrollable quantum dynamics for information processing tasks. We denote QRC as parameterless in the sense that no parametric quantum gates are incorporated into its optimization, and only classical post-processing--typically linear regression--is necessary. Therefore, it is an attractive method for finding practical quantum machine learning methods because of its NISQ compatibility, even though the necessary precision of measurements will still be affected by global measurement--induced concentration of expected values of the observables involved \cite{Cerezo2021,thanasilp2022exponential,xiong2023fundamental}.

Theoretical analyses of QRC were recently conducted in several studies \cite{chen2019dissipative, Chen_2020, Mart_nez_Pe_a_2023, gonon2023universal}. Some of these \cite{chen2019dissipative, Mart_nez_Pe_a_2021} focused on the echo state property (ESP) \cite{Jaeger2001ESP} that ensures the trainability of reservoir computers for temporal information processing tasks. However, the traditional ESP definition is not always suitable for addressing the trainability of QRC because of the non-stationarity caused by dissipations. To overcome this problem, two extensions of ESP \cite{kobayashi2024hierarchy} have been proposed that will be useful for quantum systems and other possibly nonstationary systems: nonstationary ESP and subset/subspace ESP.

Recently, there exists a work \cite{Mart_nez_Pe_a_2023} that figured out equivalent conditions for traditional ESP of QRC. Those conditions include an existence of a matrix norm for the Pauli transfer matrix (PTM) \cite{wood2015tensor} of quantum channel driving the QRC to be less than 1. To the best of our knowledge, there has been no theoretical analysis of the conditions required for quantum systems to satisfy nonstationary ESP and subset/subspace nonstationary ESPs. Independently from \cite{Mart_nez_Pe_a_2023}, we derived rather concrete sufficient conditions for the traditional and nonstationary ESP that involves spectral properties of PTM. In addition, our analysis detaches input encoding procedures from fixed reservoir dynamics, so that our theory becomes clearer for practical applications.

 One of the key conclusions of the results is that a quantum reservoir (QR) has nonstationary ESP only if it has coherent interaction with its environment, which we call \textit{coherence influx}. Coherence influx is important in following ways. First, it ensures finite output signals from the system. Second, choosing an appropriate input encoding method based on its interaction with coherence influx ensures the input dependency of the output signal. Finally, coherence influx that is not nullified by internal dynamics ensures a subspace where the spectral radius of the PTM is less than 1, which is a practical key to ensure the fading memory of the system. Typical dynamics incorporating coherence influx include the amplitude damping noise or qubit-reset operation under a compatible Hamiltonian system, and a probabilistic swap operation that is theoretically proven to provide traditional ESP in QRC \cite{Chen_2020}. For instance, in \cite{Kubota_2023, Domingo2023}, the authors suggested that amplitude-damping is a preferable type of dissipation compared with depolarization in QRC, which is explained by the fact that amplitude-damping channel has coherence influx while depolarization does not. The schematics of the necessary components in QRC, including the coherence influx, are depicted in Fig.~\ref{fig:coherence-influx}.
 
We argue that such a requirement originates from the input multiplicativity of QRC. To demonstrate this relationship to the input method in a simple way, we devised a one-dimensional multiplicative-input classical reservoir: \textit{multiplicative RC} (mRC). mRC has a direct relationship between its model parameter and information processing capability, which is quantified by memory capacity (MC) \cite{Jaeger_2001}. Furthermore, the model parameters of mRC that determines the nonstationary ESP and MC correspond to the spectral radius of PTM and the coherence influx in QRC.

We conducted numerical experiments using QRC setups with a Sherrington–Kirkpatrick (SK) Hamiltonian \cite{Sherrington_1975} with an external field to examine our theoretical results. We successfully reproduced the well-known dynamical phase transition effect \cite{_unkovi__2018,Abanin_2019, Mart_nez_Pe_a_2021} in such systems with a spectral radius of PTM. We also computed the MC of the systems and found a clear correspondence between nonstationary ESP and MC.

\begin{figure}[t]
\centering
\includegraphics[width=0.75\hsize]{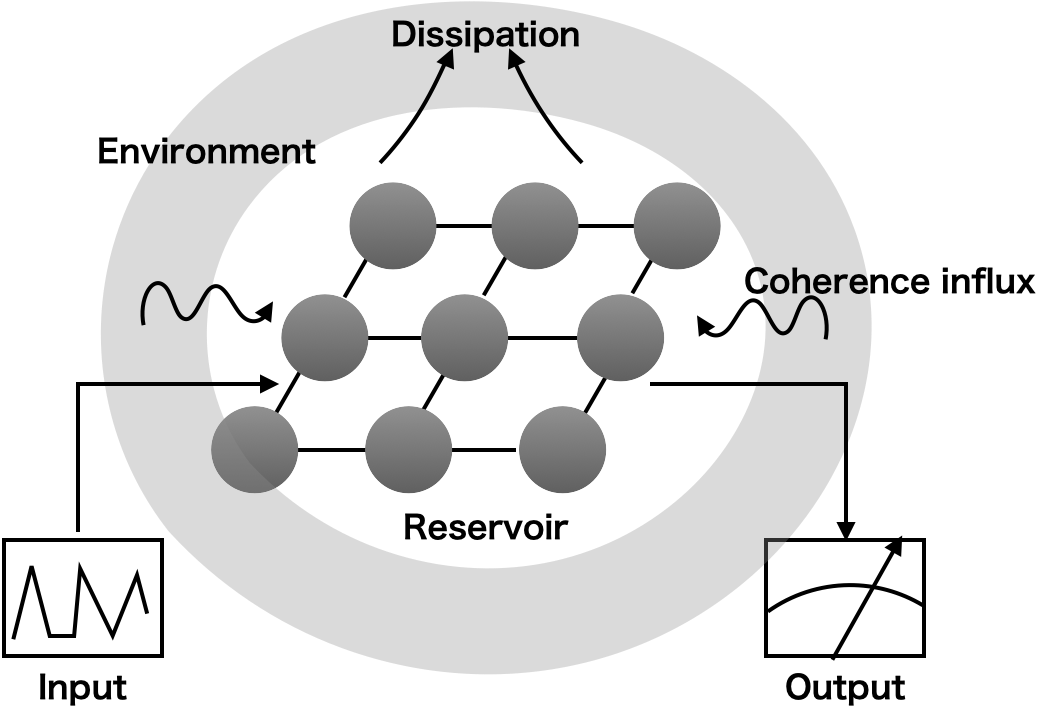}
\caption{Schematics of the components required for QRC. Coherence influx from the environment ensures a finite output signal of the QRC, while maintaining a forgetting property by forcing dissipation back to the environment.}
\label{fig:coherence-influx}
\end{figure}

Our contributions are summarized as follows:
\begin{itemize}
	\item We derived sufficient conditions for QRC to have nonstationary ESP.
	\item We theoretically proved the importance of coherence influx for fading memory.
	\item We devised multiplicative-input RC, a simple model for emphasizing the characteristics QRC.
	\item We numerically demonstrated the relationships between the spectral radius of a PTM, the nonstationary ESP, and the information processing capability.
\end{itemize}

\section{Preliminaries}
\subsection{Echo state prope sssrty}
The echo state property ensures the fading memory of reservoir dynamics regarding input and initial state dependency. Therefore, the existence of ESP ensures effective short-term memory that is required for the sequential processing of temporal information. The traditional ESP condition below that ensures the initial-state independent finite-length memory of an input history is known for stationary systems.
\begin{dfn}{Echo state property \cite{Jaeger2001ESP}}
\label{dfn:traditional-esp}

Let a compact system state space be $\mathcal{S}$ and a compact input space be $\mathcal{X}$. For an input-driven dynamical system with dynamical map $s_t = f(\{\mathbf{u}_\tau\}_{\tau < t};s_0)$ such that $f: \mathcal{S}\times \mathcal{X}^{\mathcal{T}} \to \mathcal{S}$, where $s_0$ is the initial state and $\{\mathbf{u}_\tau\}$ is a sequence of inputs indexed by time $\tau$, the ESP holds if and  only if
\begin{equation}
\label{eqn:esp}
\begin{aligned}
	\forall \{\mathbf{u}_\tau\},\ &\forall (s_0, s_0'),\\
	 &\quad \|f(\{\mathbf{u}_\tau\}_{\tau \leq t};s_0) -  f(\{\mathbf{u}_\tau\}_{\tau\leq t};s_0')\| \underset{t \to \infty} \to 0.
	\end{aligned}
\end{equation}
\end{dfn}

However, potential nonstationary systems, for example, quantum systems under uniform depolarization, cannot be well handled by traditional ESP \cite{kobayashi2024hierarchy}. An extension of the definition to potentially nonstationary systems is defined as follows \cite{kobayashi2024hierarchy}:
\begin{dfn}{Nonstationary ESP \cite{kobayashi2024hierarchy}}

Given a system dynamics $f: \mathcal{S} \times \mathcal{X}^\mathcal{T} \to \mathcal{S}$, $f$ has the \textit{nonstationary} ESP if the following condition holds:
\begin{equation}
\label{eqn:esp_ext2}
\begin{aligned}
	&\forall \{\mathbf{u}_\tau\} ,\ \forall (s_0, s_0'),\ \exists w \in \mathbb{N} < +\infty \text{ s.t. }\\
	&\quad \underset{t\to \infty}{\liminf} \left\|\mathrm{Var}_w^t[\{\mathbf{u}_\tau\}]\right\| > 0 \Rightarrow\\
	 &\quad\quad\quad\quad\frac{\left\|f(\{\mathbf{u}_\tau\}_{\tau \leq t}; s_0) -  f(\{\mathbf{u}_\tau\}_{\tau \leq t}; s_0') \right\|}{\sqrt{\min\boldsymbol{\left(}\left\|\overline{\mathrm{Var}}_w^t[f; s_0]\right\|, \left\|\overline{\mathrm{Var}}_w^t[f; s'_0]\right\|\boldsymbol{\right)}}} \underset{t \to \infty} \to 0,
	\end{aligned}
\end{equation}
where
\begin{equation}
\begin{aligned}
	\mathrm{E}_w^t[\{\mathbf{u}_\tau\}] &\equiv \frac{1}{w}\sum_{k=0}^{w-1}\mathbf{u}_{t - k},\\
	\boldsymbol{\left(}\mathrm{Var}_w^t[\{\mathbf{u}_\tau\}]\boldsymbol{\right)}_i &\equiv \mathrm{E}_w^t\boldsymbol{\left(}\{(\mathbf{u}_{\tau} - \mathrm{E}_w^t[\{\mathbf{u}_{\tau}\}])_i^2\}\boldsymbol{\right)},\\
	\overline{\mathrm{E}}_w^t[f; s_0] &\equiv \frac{1}{w}\sum_{k=0}^{w-1}f(\{\mathbf{u}_\tau\}_{\tau \leq t-k}; s_0),\\
	\left(\overline{\mathrm{Var}}_w^t[f; s_0]\right)_i &\equiv \overline{\mathrm{E}}_w^t\left[\left(f - \overline{\mathrm{E}}_w^t[f, s_0]\right)_i^2; s_0\right],\\
\end{aligned}
\end{equation}
where $(\mathbf{v})_i$ denotes the $i$-th element of a vector $\mathbf{v}$.
\end{dfn}

Here, we require output signal variety under a variable input sequence, which is ensured by the denominator part of Eq.~\eqref{eqn:esp_ext2}.

Furthermore, to cover cases where only a part of the system has ESP, subspace ESP and subset ESP are defined in \cite{kobayashi2024hierarchy}. A typical example of such a system is a two-qubit tensor product system in which only one of the qubits has ESP. 

\begin{dfn}{Subspace nonstationary ESP \cite{kobayashi2024hierarchy}}

Given system dynamics $f: \mathcal{S} \times \mathcal{X}^\mathcal{T} \to \mathcal{S}$, $f$ has \textit{subspace nonstationary ESP} if there exists a transformation $\mathbb{P}: \mathcal{S} \to \mathcal{S}'$ such that $\mathcal{S'} \subseteq \mathcal{S}$, and $\mathbb{P}\circ f$ holds nonstationary ESP.

\end{dfn}

\begin{dfn}{Subset nonstationary ESP \cite{kobayashi2024hierarchy}}

Given system dynamics $f: \mathcal{S} \times \mathcal{X}^\mathcal{T} \to \mathcal{S}$, $f$ has \textit{subset nonstationary ESP} if there exists a subset selection procedure  $\mathcal{P}: \mathcal{S} \to \mathcal{S'}$ such that $\mathcal{S'} \leq \mathcal{S}$ and $\mathcal{P}\circ f$ holds nonstationary ESP. 

	Here, the symbol $A \leq B$ denotes that $A$ is a non-void element-wise subset of $B$. For instance, if $A \equiv \mathbb{R}^m$, then $B = \mathbb{R}^n$ $(1\leq m\leq n)$. 
\end{dfn}

\subsection{Pauli transfer matrix}
The \textit{Pauli transfer matrix} \cite{Jakob_2001} representation of a quantum channel as well as a corresponding vector representation of quantum states, historically called a \textit{coherence-vector representation}, are defined below.
\begin{dfn} {Pauli transfer matrix}

Let a set of all $N$-qubit Kraus operators be $\mathcal{K}(N)$; then,
	\begin{equation}
	\label{eqn:transfer_matrix_kraus}
	\begin{aligned}
	\hat{\mathcal{O}}(N) &\equiv \bigg\{\hat{O} \in \mathbb{R}^{4^N \times 4^N} \bigg|\\&\hat{O}_{i,j} = \mathrm{tr}\bigg(P_i \sum_k K_k P_j K_k^\dagger \bigg),\ \{K_k\} \in \mathcal{K}(N)\bigg\}
	\end{aligned}
\end{equation}
is a set of all $N$-qubit PTM.

Here, Pauli strings $P_i = \bigotimes_{j=1}^N \sigma_{k_j}^{(j)}$ s.t $i = \sum_{l=0}^{N-1}k_j4^l$ and $\sigma_0^{(j)} = I^{(j)}, \sigma_1^{(j)} = X^{(j)}, \sigma_2^{(j)} = Y^{(j)}, \sigma_3^{(j)} = Z^{(j)}$ are Pauli matrices applied onto the $j$-th qubit.

We further define
\begin{equation}
	\mathcal{O}(N) \equiv \left\{ W \bigg| \begin{pmatrix}
		1 & 0\\
		\mathbf{b} & W
	\end{pmatrix} \in \hat{\mathcal{O}}(N)\right\}.
\end{equation}
\end{dfn}
A PTM can be written as 
$\hat{O} = \begin{pmatrix}
		1 & \mathbf{0}^T\\
		\mathbf{b} & W
	\end{pmatrix}$,
where $\mathbf{0} = (0, 0, \cdots, 0) \in \mathbb{R}^{4^N-1}$ and $W \in \mathbb{R}^{(4^N-1) \times (4^N-1)}$. We define the \textit{coherence influx} of $\hat{O}$ below.

\begin{dfn}{Coherence influx}\\
	The \textit{coherence influx} of a PTM $\hat{O} = \begin{pmatrix}
		1 & \mathbf{0}^T\\
		\mathbf{b} & W
	\end{pmatrix}$ is defined as $\mathbf{b}$.
\end{dfn}

We have the following simple characterization of the coherence influx.
\begin{rem}{Unital channel and coherence influx}
\label{rem:unital}

For a quantum channel $\mathcal{E}$ in PTM form: $\hat{O}_{\mathcal{E}} = \begin{pmatrix}
		1 & \mathbf{0}^T\\
		\mathbf{b} & W
	\end{pmatrix}$, $\|\mathbf{b}\| = 0$ if and only if $\mathcal{E}$ is unital; that is, $\mathcal{E}(I) = I$ in density matrix formulation.
\end{rem}

This can be induced by a probabilistic swap of quantum states with their environments, such as a local amplitude damping channel.

Given a density matrix $\rho \in \mathbb{C}^{2^N \times 2^N}$ such that $\mathrm{tr}(\rho) = 1$ and $\rho \succeq 0$, a quantum state in PTM formulation can be written as a vector $|\rho\rangle\rangle \in \mathbb{R}^{4^N}$ such that 
\begin{equation}
\label{eqn:rho_ket_ket}
	|\rho\rangle\rangle_i = \mathrm{tr}(P_i\rho).
\end{equation}
Because $\mathrm{tr}(I\rho) =  1$ for any density matrix, the first element of this vector is always a unit, so it can be written as follows:
\begin{equation}
	|\rho\rangle\rangle = \begin{pmatrix}
		1\\
		\mathbf{r}
	\end{pmatrix},
\end{equation}
where $\mathbf{r} \in \mathbb{R}^{4^N-1}$.
Let us denote a set of quantum states in PTM form as follows:
\begin{dfn}{Physical states}

	For an $N$-qubits system, we define the following set of ``physical" states:
	\begin{equation}
	\begin{aligned}
		\hat{\mathcal{Q}}(N) &\equiv \left\{(|\rho\rangle\rangle)_i = \mathrm{tr}(P_i\rho)\mid \rho^\dagger = \rho, \mathrm{tr}(\rho) = 1,\ \rho \succeq 0 \right\},\\
		\mathcal{Q}(N) &\equiv \left\{(\mathbf{r})_i = \mathrm{tr}(P_{i+1}\rho)\mid \rho^\dagger = \rho, \mathrm{tr}(\rho) = 1,\ \rho \succeq 0 \right\}.
		\end{aligned}
	\end{equation}
\end{dfn}

Several fundamental properties of physical states are proved below.
\begin{lem} {Property of a physical state}
\label{lem:mq_bloch}

\begin{enumerate}
	\item $\rho = \frac{1}{2^N}\sum_i P_i |\rho\rangle\rangle_i$.
	\item $\|\mathbf{r}\|$ does not change under unitary transformation. Therefore, 
	all pure states have the same norm.
	\item $\|\mathbf{r}\| \leq \sqrt{2^N-1} \equiv c_N$, and equality holds if and only if $\mathbf{r}$ is pure; that is, $\rho$ has an eigenvalue one.
	\item $\mathbf{r} \in \mathcal{Q}(N) \Rightarrow \forall c \leq 1,\ c\mathbf{r} \in \mathcal{Q}(N)$.
\end{enumerate}
\end{lem}

A physical state $|\rho\rangle\rangle$ follows the state transition rule under $\hat{O}$, as follows:
\begin{equation}
	|\rho'\rangle\rangle = \hat{O}|\rho\rangle\rangle\\,
\end{equation}
or equivalently,
\begin{equation}	
	\begin{pmatrix}
		1\\
		\mathbf{r'}
	\end{pmatrix} = \begin{pmatrix}
		1\\
		\mathbf{b} + W\mathbf{r}
	\end{pmatrix}.
\end{equation}

Suppose that we have an input signal $\mathbf{u}$ and an input-dependent ``encoding" PTM $\hat{R}(\mathbf{u}) = \begin{pmatrix}
	1 & \mathbf{0}^T\\
	\mathbf{0} & R(\mathbf{u})
\end{pmatrix}$. Then, the overall state update can be written as
\begin{equation}
\begin{aligned}
	\begin{pmatrix}
		1\\
		\mathbf{r'}
	\end{pmatrix} &= \hat{O}\hat{R}(\mathbf{u})\begin{pmatrix}
		1\\
		\mathbf{r}
	\end{pmatrix}\\
	&=\begin{pmatrix}
		1\\
		 WR(\mathbf{u})\mathbf{r} + \mathbf{b}
	\end{pmatrix}.
	\end{aligned}
\end{equation}
That is,
\begin{equation}
\label{eqn:qrc_state_update}
	\mathbf{r}' =  WR(\mathbf{u})\mathbf{r} + \mathbf{b}.
\end{equation}

Here, we also define the input encoding methods of QRC in PTM form. 

In general, given an input sequence $\{\mathbf{u}_t \in \mathcal{X}\}_{t \in \mathcal{T}}$ with countable $\mathcal{T}$, the input-driven system dynamics governed by input encoding $\hat{E}(\mathbf{u})= \begin{pmatrix}
	1 & \mathbf{0}^T\\
	\mathbf{a}(\mathbf{u}) & E(\mathbf{u})
\end{pmatrix}$ can be written as follows:
\begin{equation}
\begin{aligned}
	\mathbf{r}^{(t+1)} &= \mathbf{b} + W\mathbf{a}(\mathbf{u}_t) + WE(\mathbf{u}_t)\mathbf{r}^{(t)},
\end{aligned}
\end{equation}
where $\mathbf{a}: \mathcal{X} \to \mathbb{R}^{4^N-1}$ and $E: \mathcal{X}\to \mathbb{R}^{\left(4^N-1\right) \times \left(4^N-1\right)}$ are the input-dependent vector and matrix, respectively. Specifically, \textit{unitary} input encoding is defined below.
\begin{dfn}{Unitary input encoding}

The \textit{unitary input encoding} for an input sequence composed of $\mathcal{X}\equiv \mathbb{R}^d$ inputs at each step, is a parametric unitary $\hat{R} \equiv \begin{pmatrix}
 1 & 0\\
 0 & R	
 \end{pmatrix}$ such that $R: \mathcal{X} \to \mathbb{R}^{\left(4^N-1\right) \times \left(4^N-1\right)}$.
\end{dfn}

Therefore, the general form of the system state under an input-driven dynamics with unitary input encoding is
\begin{equation}
\label{eqn:qrc_gen_form}
\begin{aligned}
	\mathbf{r}^{(t)} &= \sum^{t-1}_{\tau=0}\left(\prod_{1\leq n \leq \tau}WR(\mathbf{u}_{t - n})\right)\mathbf{b} 
	+ \left(\prod_{\tau \leq t}WR(\mathbf{u}_{\tau})\right)\mathbf{r}^{(0)}.
	\end{aligned}
\end{equation}

In \cite{Mart_nez_Pe_a_2023}, it was proved that QRC driven by Eq.~\eqref{eqn:qrc_gen_form} has \textit{traditional} ESP if and only if there exists a sub-multiplicaive matrix norm $\|\cdot\|$ such that $\|WR(\mathbf{u})\| < 1$ for all $\mathbf{u} \in \mathcal{X}$, where $R$ is continuous, and eventually $\{R(\mathbf{u}) \mid \mathbf{u} \in \mathcal{X}\}$ forms a compact set provided that $\mathcal{X}$ is compact. If, for instance, such norm is the spectral norm or the Frobenius norm, thenm a simpler equivalent condition can be obtained in our notation. That is, $\|W\| < 1$, because $R(\mathbf{u})$ is an orthogonal transform, and does not change these norms. It should be noted that the compactness condition is only used for necessity of the existence of such norm for the traditional ESP, as we can see in the proof of Thm.~2.19 and Cor.~6.4 in \cite{hartfiel2002nonhomogeneous}, and only boundedness of $R(\mathbf{u})$s are required for the sufficiency. Because the space of all CPTP maps is bounded, even if $\mathcal{X}$ is not compact nor $R$ is not continuous, the existence of the norm is sufficient for the traditional ESP. Following this fact, we do not care about the continuity of $R$ nor the compactness of the input space $\mathcal{X}$ when dealing with the nonstationary ESP in this manuscript.

\subsection{Memory capacity and information processing capacity}
Memory capacity \cite{Jaeger_2001} and information processing capacity (IPC) \cite{dambre2012} quantify linear and non-linear input dependency of the output sequence, respectively, in input-driven dynamical systems by only using the input sequence and state sequence. Suppose we have collected state sequence $\{x_t\}$ under inputs of $\{u_t\}$; then, the total memory capacity $C_{tot}^{MC}$ can be calculated as
\begin{equation}
\begin{aligned}
	C_{tot}^{MC}(\{u_t\}, \{x_t\}) &= \sum_k C_{k}^{MC}(\{u_t\}, \{x_t\}), \text{ where}\\
	C_{k}^{MC}(\{u_t\}, \{x_t\}) &= \frac{\mathbb{E}_t[u_{t-k}x_t^T] \mathbb{E}_t[x_tx_t^T] \mathbb{E}_t[u_{t-k}x_t]}{\mathbb{E}_t[\left(u_t - \mathbb{E}_t[u_t]\right)^2]}.
	\end{aligned}
\end{equation}

The calculation of IPC requires a set of orthogonal functions that non-linearly transform the input sequence $\{u_t\}$. Given a set of orthogonal functions $\mathcal{Y}$, the total information processing capacity $C_{tot}^{IPC}$ can be calculated as
\begin{equation}
\begin{aligned}
	C_{tot}^{IPC}&\left(\{u_t\}, \{x_t\};\mathcal{Y}\right) = \sum_{d}C_{d}^{IPC}\left(\{u_t\}, \{x_t\}; \mathcal{Y}\right), \text{ where}\\
	C_{d}^{IPC}&\left(\{u_t\}, \{x_t\}; \mathcal{Y}\right) = \\
	&\sum_{\{v_t\} \in \mathcal{Y}\left(\{u_t\}; d\right)}
	\frac{\mathbb{E}_t[v_{t}x_t^T] \mathbb{E}_t[x_tx_t^T] \mathbb{E}_t[v_{t}x_t]}{\mathbb{E}_t[\left(v_t - \mathbb{E}_t[v_t]\right)^2]},
\end{aligned}
\end{equation}
where $\mathcal{Y}\left(\{u_t\}; d\right)$ is a set of degree $d$ non-linear transformed sequence of inputs. An element of $\mathcal{Y}\left(\{u_t\}; d\right)$: $\{v_t\}$ can be calculated using degree $d_i$ functions $Y_{d_i} \in \mathcal{Y}$ as
\begin{equation}
\label{eqn:v_t}
	v_t = \prod_i Y_{d_i}(u_{t-k_i})\quad \text{s.t.}\quad \sum_{i}d_i = d.
\end{equation}
Here, each argument of $Y_{d_i}$ on the right-hand side of Eq.~\eqref{eqn:v_t} involves delays $t_i$. Ideally, every combinations of $\{t_i \in [0, \infty)\}$ and $\{Y_{d_i} \mid  \sum_{i}d_i = d\}$ must be used to obtain the full spectrum of $C_{d}^{IPC}$. However, only a selected subset of them are used in real calculations because of limited computational resources.

\section{Main results}
For the rest of the manuscript, we denote the spectral norm a matrix $A$, that is, $\sup_{\mathbf{v}}\frac{\|A\mathbf{v}\|_2}{\|\mathbf{v}\|_2}$, as $\sigma_{\mathrm{max}}(A)$, and the spectral radius of a matrix $A$, that is, $\max_{i}|\lambda_i(A)|$, as $\rho(A)$, where $\lambda_i(A)$ is the $i$-the eigenvalue of $A$.
\subsection{Necessity of coherence influx}
\label{subsec:esp}
\subsubsection{Theoretical result}
First, we show a simplified sufficient condition for the traditional ESP of QRC defined in Def.~\ref{dfn:traditional-esp}.
\begin{lem}{Sufficient condition for traditional ESP}
\label{lem:suffice_trad_esp}

Suppose that we have a QRC with a PTM $\hat{O} = \begin{pmatrix}
	1 & \mathbf{0}^T\\
	\mathbf{b} & W
\end{pmatrix} \in \hat{\mathcal{O}}(N)$ which is driven under a unitary input encoding $\hat{R} = \begin{pmatrix}
		1 & \mathbf{0}^T\\
		\mathbf{0} & R
	\end{pmatrix}: \mathcal{X} \to \hat{\mathcal{O}}(N)$.
In addition, let
\begin{equation}
	s_t\left(W, R; \{\mathbf{u}_t\}\right) \equiv \sigma_{\mathrm{max}}\left(\prod_t WR(\mathbf{u}_t)\right),
\end{equation}
for an input sequence $\{\mathbf{u}_t \in \mathcal{X}\}$, then, the QRC has traditional ESP if 
		\begin{equation}
		\label{eqn:spectral_converge}
			s_t\left(W, R; \{\mathbf{u}_t\}\right) \underset{t\to \infty}\to 0,
		\end{equation} holds for all $\{\mathbf{u}_t\} \in \mathcal{X}^\mathcal{T}$.
\end{lem}

The proposition below shows the importance of the coherence influx for QRC's nonstationary ESP. In other words, the unital channel does not have the nonstationary ESP under unitary input encoding.

\begin{prop}{Coherence influx is indispensable}
\label{prop:coherence-influx}
Suppose that we have a PTM $\hat{O} = \begin{pmatrix}
	1 & \mathbf{0}^T\\
	\mathbf{b} & W
\end{pmatrix} \in \hat{\mathcal{O}}(N)$. If $\|\mathbf{b}\|= 0$, the nonstationary ESP does not hold under any unitary input encoding.
\end{prop}

From Rem.~\ref{rem:unital}, Prop.~\ref{prop:coherence-influx} indicates that no unital channel can be used as a dynamics of QRC under unitary input encoding. We need a formal definition of an \textit{injective} map for our theorems.

\begin{dfn}{Injective map}
	
	Let $\mathcal{X}$ and $\mathcal{Y}$ be metric spaces with their respective distance function denoted as $d(\cdot, \cdot)$. A function $f: \mathcal{X} \to \mathcal{Y}$ is called injective if and only if for any positive real $\delta> 0$ and $u,v\in \mathcal{X}$ such that $d(u, v) > \delta$, there exists a positive real $\epsilon > 0$ such that $d(f(u), f(v))> \epsilon$. 
\end{dfn}
We have the following sufficient condition for the nonstationary ESP using the definition above:

\begin{thm}{Sufficient condition for the nonstationary ESP of QRC}

\label{thm:sufficient_ns_esp}
A QRC with a PTM $\hat{O} = \begin{pmatrix}
	1 & \mathbf{0}^T\\
	\mathbf{b} & W
\end{pmatrix} \in \hat{\mathcal{O}}(N)$ has the nonstationary ESP under a unitary input encoding $\hat{R} = \begin{pmatrix}
		1 & \mathbf{0}^T\\
		\mathbf{0} & R
	\end{pmatrix}: \mathcal{X} \to \hat{\mathcal{O}}(N)$ if all of the following conditions hold:
	\begin{enumerate}
		\item Inverse matrices $\mathcal{G}^{-1}(\mathbf{u}_t)\equiv \left(I - WR(\mathbf{u}_t)\right)^{-1}$ always exists for all $\mathbf{u}_t \in \mathcal{X}$.
		\item $\mathbf{u}_t \mapsto \mathcal{G}^{-1}(\mathbf{u}_t)\mathbf{b}$ is an injective map from $\mathcal{X}$ to $\mathbb{R}^{4^N-1}$.
		\item Eq.~\eqref{eqn:spectral_converge} holds.
		
	\end{enumerate}
\end{thm}

The equivalent condition for the existence of $\mathcal{G}^{-1}(\mathbf{u}_t)$ is that no eigenvalue of $WR(\mathbf{u}_t)$ is equal to 1 for all $\mathbf{u}_t \in \mathcal{X}$. In addition, the typical conditions necessary for $\mathbf{u}_t \mapsto \mathcal{G}^{-1}(\mathbf{u}_t)\mathbf{b} : \mathcal{X} \to \mathbb{R}^{4^N-1}$ to be injective are summarized in Rem.~\ref{rem:typical_ns_esp}.

\begin{rem}{Typical conditions necessary for injectiveness}

The following are the typical conditions necessary for $\mathbf{u}_t \mapsto \mathcal{G}^{-1}(\mathbf{u}_t)\mathbf{b}$ to be injective as a map from $\mathcal{X}$ to $\mathbb{R}^{4^N-1}$: 
\begin{enumerate}
	\item There exists a real positive $\delta$ such that $\|\mathbf{b}\| > \delta$. For instance, an amplitude-damping channel or a qubit-reset operation exists.
	\item $R(\mathbf{u}_t): \mathcal{X} \to \mathcal{O}(N)$ is injective. For instance, it is a rotation around a fixed axis on $\mathbb{R}^{4^N-1}$ whose degree of rotation is proportional to the input $\mathbf{u}_t$.
	\item $\mathbf{b}$ is not an eigenvector of $R(\mathbf{u}_t)$ for any $\mathbf{u}_t \in \mathcal{X}$. This includes cases where $\mathbf{b}$ has zero entries only on dimensions where $R(\mathbf{u}_t)$ applies. For example, $R$ is not a local $R_Z$ when $\mathbf{b}$ comes only from local amplitude damping.
	\item $W$ does not nullify the input encoding $R(\mathbf{u}_t)$. For example, it does not have zero entries that void the $\mathbf{u}_t$ dependency of $R(\mathbf{u}_t)$. A counter-example is a case in which $R$ is applied to a subset of qubits and $\hat{O}$ applies complete depolarization on that set of qubits.
\end{enumerate}
\label{rem:typical_ns_esp}
\end{rem}

As we can see from Rem.~\ref{rem:typical_ns_esp}, under typical rotational unitary input encoding $R$ that is compatible with the system dynamics $W$, finite coherence influx $\mathbf{b}$ that can be modulated by $R$ is the key for nonstationary ESP.

We have sufficient conditions for Eq.~\eqref{eqn:spectral_converge} to hold. First, we need the following definition of Schur stability \cite{yildiz2012re}:

\begin{dfn}{Schur stability \cite{yildiz2012re}}

	Suppose that we have a matrix $W \in \mathbb{R}^{N \times N}$; then, $W$ is Schur stable if there exists a symmetric matrix $P \succ 0$ such that $W^TPW - P \prec 0$. Here, $A \succ 0$ and $A \prec 0$ means positive definiteness and negative definiteness of a matrix $A$, respectively.
\end{dfn}

Based on the following monotonicity condition of Hilbert--Schmidt distance \cite{Wang_2009}, sufficient conditions for Eq.~\eqref{eqn:spectral_converge} can be written:

\begin{lem}{Strict contraction of Hilbert--Schmidt distance}
\label{lem:monotonic_hs_dist}
	
Suppose that we have a PTM $\hat{O} = \begin{pmatrix}
	1 & \mathbf{0}^T\\
	\mathbf{b} & W
\end{pmatrix} \in \hat{\mathcal{O}}(N)$ and a unitary input encoding $\hat{R} = \begin{pmatrix}
		1 & \mathbf{0}^T\\
		\mathbf{0} & R
	\end{pmatrix}: \mathcal{X} \to \hat{\mathcal{O}}(N)$. For any initial states $\mathbf{r}^{(0)}$ and $\mathbf{r}^{'(0)}$, let the Euclidean norm between these states at time $t$, which is equivalent to the Hilbert--Schmidt distance between corresponding density matrices as follows:
	\begin{equation}
			\|\Delta\mathbf{r}^{(t)}\| \equiv \|\mathbf{r}^{(t)}- \mathbf{r}^{'(t)}\|.
	\end{equation}
Then, the following statements are true:
\begin{enumerate}
	\item $\|\Delta\mathbf{r}^{(t)}\|$ is strictly decreasing with respect to $t$ if and only if $\mathcal{G}(\mathbf{u}_t) + \mathcal{G}(\mathbf{u}_t)^T$ is positive definite for every $\mathbf{u}_t \in \mathcal{X}$.
	\item $\|\Delta\mathbf{r}^{(t)}\|$ is strictly decreasing with respect to $t$ if $WR(\mathbf{u}_t)$ does not have an eigenvalue 1 and $\mathcal{G}(\mathbf{u}_t)$ is diagonalizable for every $\mathbf{u}_t \in \mathcal{X}$.
	\item $\|\Delta\mathbf{r}^{(t)}\|$ is strictly decreasing with respect to $t$ if there exists a positive symmetric matrix $P$ such that $WR(\mathbf{u}_t)$ is Schur stable with respect to $P$ for every $\mathbf{u}_t \in \mathcal{X}$.
\end{enumerate}
\end{lem}
Here, condition 3 is similar to the sufficient condition of the traditional ESP of classical echo state network (ESN) \cite{yildiz2012re}, a reservoir based on the artificial neural network framework \cite{Jaeger2001ESP}, where diagonally Schur stability of recurrent weight implies the traditional ESP. This result indicates that techniques used in the analysis of classical reservoir computing will be useful in the analysis of QRC when it is represented by the PTM form. 

Another sufficient condition for Eq.~\eqref{eqn:spectral_converge} is the existence of a matrix norm $\|\cdot\|$ such that $\|WR(\mathbf{u})\| < 1$ for all $\mathbf{u} \in \mathcal{X}$ \cite{Mart_nez_Pe_a_2023}.

\begin{prop}{Sufficient condition for the convergence of spectral norms}
\label{prop:suffice_convergence}

Suppose that we have a PTM $\hat{O} = \begin{pmatrix}
	1 & \mathbf{0}^T\\
	\mathbf{b} & W
\end{pmatrix} \in \hat{\mathcal{O}}(N)$ and a unitary input encoding $\hat{R} = \begin{pmatrix}
		1 & \mathbf{0}^T\\
		\mathbf{0} & R
	\end{pmatrix}: \mathcal{X} \to \hat{\mathcal{O}}(N)$. Then, Eq.~\eqref{eqn:spectral_converge} holds if the following conditions hold:
	
	\begin{enumerate}
		\item $\mathcal{G}(\mathbf{u}_t) + \mathcal{G}(\mathbf{u}_t)^T$ is positive definite for every $\mathbf{u}_t \in \mathcal{X}$.
		\item $\mathrm{Ker}\boldsymbol{(}\mathcal{G}(\mathbf{u}_t)\boldsymbol{)} = \emptyset$ and $\mathcal{G}(\mathbf{u}_t)$ is diagonalizable for every $\mathbf{u}_t \in \mathcal{X}$.
		\item There exists a positive symmetric matrix $P$ such that $WR(\mathbf{u}_t)$ is Schur stable with respect to $P$ for every $\mathbf{u}_t \in \mathcal{X}$.
 		\item There exists a matrix norm $\|\cdot\|$ such that $\|WR(\mathbf{u})\| < 1$ for every $\mathbf{u} \in \mathcal{X}$.
	\end{enumerate}
	
\end{prop}

However, we argue that finding a matrix norm that satisfies condition 4 is not simple. For instance, when evaluating the spectral norm $\sigma_{\mathrm{max}}\boldsymbol{(}WR(\mathbf{u})\boldsymbol{)} = \sigma_{\mathrm{max}}(W)$, it becomes clear that it is too strict, because applying the local reset unitary input encoding below, which is a typical protocol of input encoding in QRC, always makes $\sigma_{\mathrm{max}}(W)$ larger than 1, yet such QRC is successful in producing temporal information processing capabilities.

\begin{dfn}{Local reset unitary input encoding}
\label{dfn:reset-input}

    Given a set of qubits $K_{in}$ for input targets, a parametric unitary after complete amplitude damping as
    
   \label{dfn:reset-encoding}
    \begin{equation}
        \hat{O}_{reset}(\mathbf{u}) \equiv \bigotimes_{i \notin K_{in}} I^{(i)} \otimes \bigotimes_{i \in K_{in}} \left(\hat{R}(\mathbf{u})^{(i)} \hat{\Gamma}(1)^{(i)}\right),
    \end{equation}
    where
    \begin{equation}
    	\hat{\Gamma}(\gamma) = \begin{pmatrix}
    		1 & 0 & 0 & 0\\
    		0 & \sqrt{1 - \gamma} & 0 & 0\\
    		0 & 0 & \sqrt{1 - \gamma} & 0\\
    		\gamma & 0 & 0 & 1 - \gamma\\
    	\end{pmatrix},
    \end{equation}
    is called a \textit{local reset unitary input encoding}.
\end{dfn}
A QRC under a local reset unitary input encoding has the property below.
\begin{rem}{Spectral norm under local reset unitary input encoding}

\label{rem:spectral_norm_reset}
$N$-qubits QRC with $M<N$-qubits local reset unitary input encoding (Def.~\ref{dfn:reset-input}) without any dissipations except those from reset operations has $\sigma_{\mathrm{max}}(W) = 2^\frac{M}{2}$.
\end{rem}
Rem.~\ref{rem:spectral_norm_reset} implies that $\sigma_{\mathrm{max}}(W) > 1$ for every $1 \leq M < N$-qubits reset unitary input encoding under unitary dynamics, which means that the spectral norm of $W$ does not have any information about ESP. In addition, the fact that Frobenius norm $\|W\|_F = \sqrt{\sum_i \sigma_i(W)^2}$ is lower bounded by the spectral norm implies that Frobenius norm is also uninformative about the ESP of QRC.

It should be noted that the condition 1 and 2 of Prop.~\ref{prop:coherence-influx}, that is, positive definiteness of $\mathcal{G}(\mathbf{u}_t) + \mathcal{G}(\mathbf{u}_t)^T$ and $\mathrm{Ker}\boldsymbol{(}\mathcal{G}(\mathbf{u}_t)\boldsymbol{)} = \emptyset$ both imply $\rho\boldsymbol{(}WR(\mathbf{u}_t)\boldsymbol{)} \neq 1$. Therefore, we expect that there will be a large dependence of the nonstationary ESP of QRC on the spectral radius $\rho(W)$ and $\rho(WR(\mathbf{u}_t))$, which are lower bounds of all matrix norms of $W$ and $WR(\mathbf{u}_t)$, respectively. This is similar to the case in the ESN \cite{Basterrech2017}, where the spectral radius of recurrent weight is found to be practically important for the traditional ESP. We define the averaged spectral radius of $WR(\mathbf{u}_t)$ as follows:

\begin{dfn}{Effective input-driven spectral radius}\\

For each input sequence $\{\mathbf{u}_\tau \in \mathcal{X}\}_{\tau \in \mathcal{T}}$ and input encoding $R: \mathcal{X} \times \mathcal{Q}(N) \to \mathcal{Q}(N)$, the \textit{effective input-driven spectral radius} of $W$ is
\begin{equation}
	\rho_{\mathrm{eff}}(W; R,\{\mathbf{u}_{\tau}\}) \equiv \mathbb{E}_\tau^{g}[\rho\boldsymbol{(}WR(\mathbf{u_\tau})\boldsymbol{)}],
\end{equation}
where $\mathbb{E}^{g}$ denotes a geometric mean.
\end{dfn}
It should be noted that if $\mathbb{E}_{\{\mathbf{u}_t\}}[\rho_{\mathrm{eff}}\boldsymbol{(}W; R,\{\mathbf{u}_{\tau}\}\boldsymbol{)}] < 1$, Eq.~\eqref{eqn:spectral_converge} is practically satisfied.

We expect the conjecture below.
\begin{conj}{Existence of unitary input encoding for the nonstationary ESP}
\label{conj:exist_unitary}

If $\rho(W) < 1$, then, there always exists a non-trivial unitary input encoding $\hat{R}$ such that 

Eq.~\eqref{eqn:spectral_converge} is satisfied. Furthermore, the convergence rate with respect to $t$ is $O\boldsymbol{(}\rho_{\mathrm{eff}}(W; R,\{\mathbf{u}_{\tau}\})^t\boldsymbol{)}$.
\end{conj}
Here, a non-trivial unitary input encoding means that it is not always an identity. If it is always identity, then $\rho(W^t) = \rho(W)^t \underset{t\to \infty} \to 0$ implies Eq.~\eqref{eqn:spectral_converge}. In addition, the convergence rate is expected to be approximately $O\boldsymbol{(}\rho(W)^t\boldsymbol{)}$.

\subsubsection{Numerical result}

\begin{figure}[t]
		\centering
		\includegraphics[width=0.9\hsize]{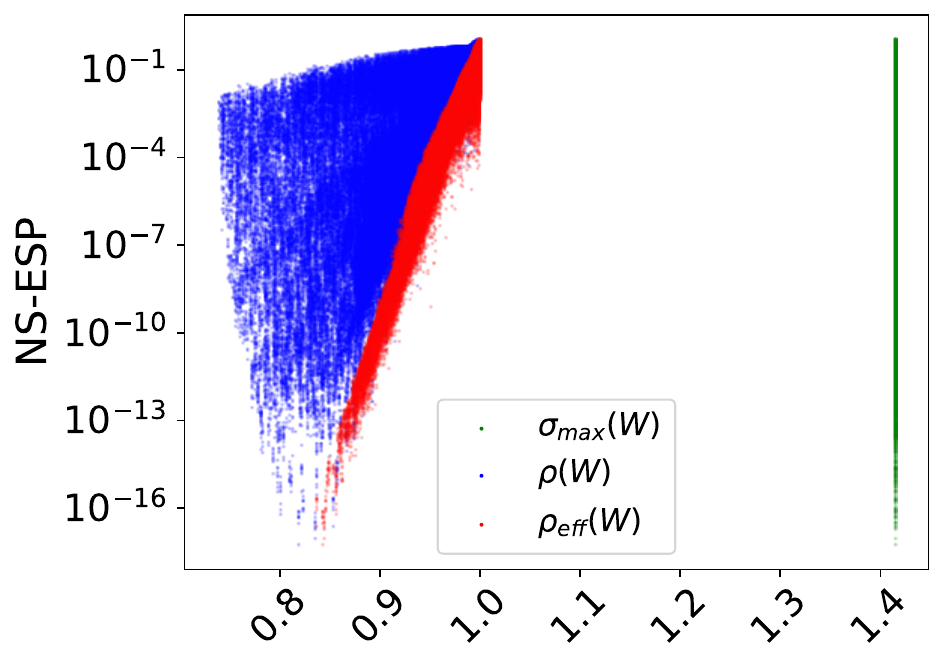}
\caption{Nonstationary ESP indicator (NS-ESP) calculated by Eq.~\eqref{eqn:ns_esp_indicator}. The horizontal axis for each color is noted in the legend. The dependency of the nonstationary ESP indicator values for the spectral radius of $W$: $\rho(W)$ (blue), effective spectral radius of $W$: $\rho_{\mathrm{eff}}(W)$ (red), and the spectral norm of $W$: $\sigma_{\mathrm{max}}(W)$ (green) are plotted. The experiment setups are described in Sec.~\ref{subsec:esp}. $|\{\mathbf{u}_t\}| = 200$ and $w=20$ are used for calculating nonstationary ESP indicator.}
		\label{fig:sigma_vs_rho_sk}
\end{figure}
To demonstrate numerically the system's ESP dependency to $\rho(W)$, $\rho_{\mathrm{eff}}(W)$, and $\sigma_{\mathrm{max}}(W)$, we calculated the nonstationary ESP indicator defined as follows:
\begin{equation}
\label{eqn:ns_esp_indicator}
\begin{aligned}
	&\mathcal{I}_{NS}(t, s_0, s_0', w) \equiv \\
	&\quad \mathcal{I}_{ESP}(t, s_0, s_0') \times \frac{\sqrt{\min\left(\left\|\overline{\mathrm{Var}}_w^w[f; s_0]\right\|, \left\|\overline{\mathrm{Var}}_w^w[f; s'_0]\right\|\right)}}{\sqrt{\min\left(\left\|\overline{\mathrm{Var}}_w^t[f; s_0]\right\|, \left\|\overline{\mathrm{Var}}_w^t[f; s'_0]\right\|\right)}}\\
	&\text{where}\\
	&\mathcal{I}_{ESP}(t, s_0, s_0') \equiv \frac{\|f(\{\mathbf{u}_\tau\}_{\tau \leq t};s_0) -  f(\{\mathbf{u}_\tau\}_{\tau\leq t};s_0')\|}{\|s_0 -s_0'\|},
	\end{aligned}
\end{equation}
using two-qubit QRC governed by the following SK-type Hamiltonian:
\begin{equation}
\label{eqn:sk_hamiltonian}
	H=\sum_{i>j=1}^N J_{i j} \sigma_i^x \sigma_j^x+\frac{1}{2} \sum_{i=1}^N\left(h+D_i\right) \sigma_i^z,
\end{equation}
where

\begin{equation}
\label{eqn:sk_sampling}
\begin{aligned}
 J_{ij} &\sim \mathrm{Uniform}([-J_s/2, J_s/2]),\\
 D_i &\sim \mathrm{Uniform}([-WJ_s/2, WJ_s/2]),
\end{aligned}
\end{equation}
and the reset-input encoding that is defined in Def.~\ref{dfn:reset-encoding}.
The parametric unitary used in our reset unitary input encoding is
\begin{equation}
\label{eqn:reset_unitary_input_sk}
	\hat{R}(\mathbf{u})^{(i)} = \begin{pmatrix}
		1 & \mathbf{0}^T\\
		\mathbf{0} & R_y\left(\arccos\left(\mathbf{u}\right)\right)
	\end{pmatrix}.
\end{equation}

Ten Hamiltonians were sampled for each configuration in $J_s, K \in \{10^{-2 + 4k/99} \mid k\in \{0,1,\dots,99\}\}$, resulting in a total of 50,000 random SK Hamiltonians generated for this calculation. The results are depicted in Fig.~\ref{fig:sigma_vs_rho_sk}. We can observe that, consistent with Conj.~\ref{conj:exist_unitary} $\rho_{\mathrm{eff}}(W; R, \{\mathbf{u}_t\})$ has a stronger log-linear relationship with the nonstationary ESP indicator than $\rho(W)$ itself. However, it should be noted that the nonstationary ESP is always satisfied if $\rho(W) < 1$ in our setup. In addition, $\sigma_{\mathrm{max}}(W)$ has no relation to the nonstationary ESP indicators, as suggested by Rem.~\ref{rem:spectral_norm_reset}.

\subsection{Non-vanishing coherence influx}
\subsubsection{Theoretical results}
We also have the following theorem:
\begin{lem}{Positivity-ensured subspace under non-vanishing coherence influx}
\label{lem:positiveity_subspace}

Suppose that we have a PTM $\hat{O} = \begin{pmatrix}
	1 & \mathbf{0}^T\\
	\mathbf{b} & W
\end{pmatrix} \in \hat{O}(N)$ such that $W$ is diagonalizable. If $\|W\mathbf{b}\| > 0$, then, there exists an invariant subspace $\mathcal{Q}_\mathbf{b}$ of $W$ such that $\mathbf{b} \in \mathcal{Q}_\mathbf{b}$ and $\vec\lambda \notin \mathcal{Q}_\mathbf{b}$ for all eigenvectors $\vec\lambda$ of $W$ that has corresponding eigenvalue of 1. That is, $W$ does not have an eigenvalue 1 in $\mathcal{Q}_\mathbf{b}$; hence, $(I - W)^T + (I - W)$ is positive definite in $\mathcal{Q}_\mathbf{b}$ provided that $W$ is diagonalizable. That is, $\mathbb{P}_{\mathcal{Q}_\mathbf{b}}\left[(I - W)^T + (I - W)\right]$ is positive definite where $\mathbb{P}_{\mathcal{Q}_\mathbf{b}}$ is a projector from $\mathcal{Q}(N)$ onto $\mathcal{Q}_\mathbf{b}$. Furthermore, $\mathcal{Q}_b^{\perp} \oplus \mathcal{Q}_b = \mathcal{Q}(N)$, where $\mathcal{Q}_b^{\perp}= \mathrm{span}(\{\vec\lambda(W) \mid \lambda(W) = 1\})$. Here, $\vec\lambda(W)$ and $\lambda(W)$ denote an eigenvector and the corresponding eigenvalue of $W$, respectively.
	
\end{lem}

\begin{cor}{Subspace nonstationary ESP under non-vanishing coherence influx}
\label{cor:subspace_ns_esp_suffice}

Suppose that we have a PTM $\hat{O} = \begin{pmatrix}
	1 & \mathbf{0}^T\\
	\mathbf{b} & W
\end{pmatrix} \in \hat{\mathcal{O}}(N)$ such that $\|W\mathbf{b}\| > 0$.

Let the projector onto $\mathcal{Q}_\mathbf{b}$ be $\mathbb{P}_{\mathcal{Q}_\mathbf{b}}$. If there exists an input encoding $\hat{R}$ such that $\mathcal{Q}_\mathbf{b}$ and $\mathcal{Q}_\mathbf{b}^\perp$ are its invariant subspaces, then, $\hat{O}$ has the subspace nonstationary in $\mathcal{Q}_\mathbf{b}$ under $\hat{R}$ if the following conditions hold:
 \begin{enumerate}
 	\item Inverse matrices $[\mathcal{G}_\mathbf{b}(\mathbf{u}_t)]_{\mathcal{Q}_\mathbf{b}}^{-1}$ always exist for all $\mathbf{u}_t \in \mathcal{X}$, where $\mathcal{G}_\mathbf{b}(\mathbf{u}_t) \equiv \mathbb{P}_{\mathcal{Q}_\mathbf{b}}\left(I - WR(\mathbf{u}_t)\right)$ and $[\cdot]_{{\mathcal{Q}_\mathbf{b}}}$ denotes a representation of a vector or a matrix in $\mathcal{Q}_\mathbf{b}$ with its orthonormal basis set.
 	\item $\mathbf{u}_t \mapsto \mathcal{G}_\mathbf{b}^{-1}(\mathbf{u}_t)\mathbf{b}$ is an injective map from $\mathcal{X}$ to $\mathcal{Q}_\mathbf{b}$.
 	\item $s_t^{\mathbf{b}}\left(W, R; \{\mathbf{u}_t\}\right) \equiv \sigma_{\mathrm{max}}\left(\mathbb{P}_{\mathcal{Q}_\mathbf{b}}\prod_t WR(\mathbf{u}_t)\right) \underset{t\to \infty}\to 0$ for any $\{\mathbf{u}_t \in \mathcal{X}\}$.
 \end{enumerate}
\end{cor}
An important indication here is that even if $s_t$ does not converge, we can use QRC by projecting outputs to $\mathcal{Q}_\mathbf{b}$ if $s_t^\mathbf{b}$ converges. Furthermore, $\mathbf{b}\in \mathcal{Q}_\mathbf{b}$ is consistent with the existence of an injective map $\mathbf{u}_t \mapsto \mathcal{G}_\mathbf{b}^{-1}(\mathbf{u}_t)\mathbf{b}$.
\subsubsection{Numerical results}
\label{subsubsec:wb_numerical}
\begin{figure}[t]
	\centering
	\includegraphics[width=0.9\hsize]{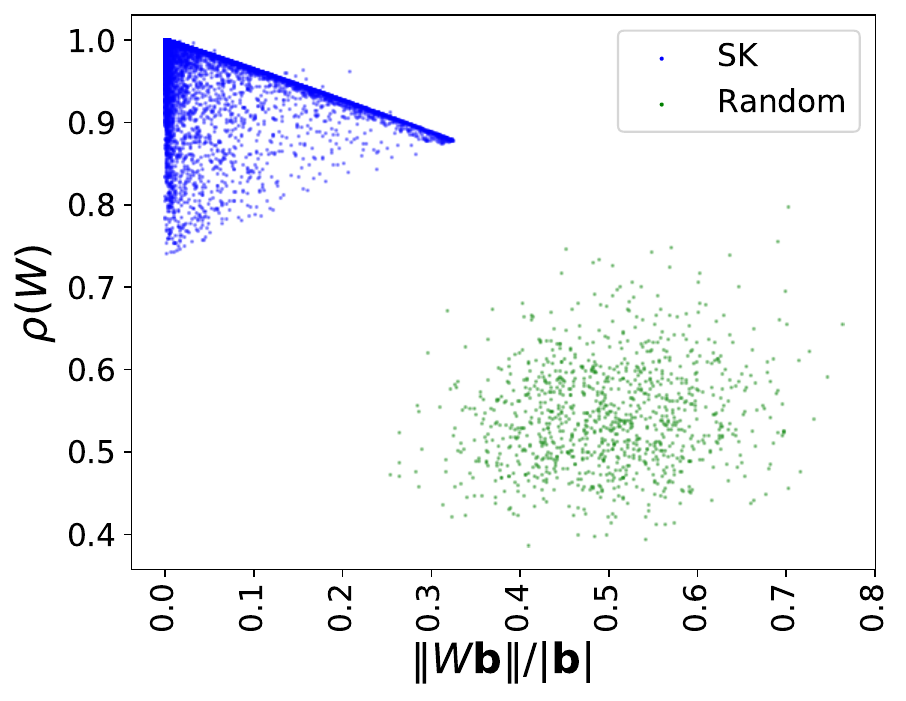}
	\includegraphics[width=0.9\hsize]{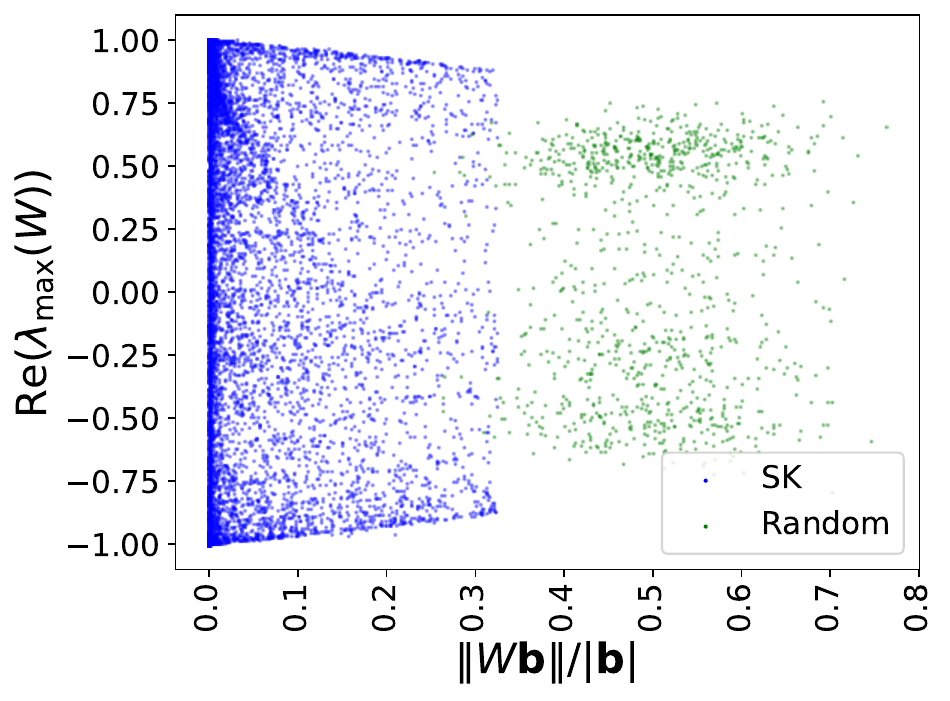}
	\caption{ Distribution of (a): $\rho(W)$ and (b): $\mathrm{Re}(\lambda_{\max}(W))$ to $\|W\mathbf{b}\|/\|\mathbf{b}\|$ in randomly sampled $2$-qubit QRC with reset-input encoding by $R_Y$, as described in Sec.~\ref{subsec:esp}. Blue dots (SK) use the same setup as in Sec.~\ref{subsec:esp}. The hamiltonians used in calculating green dots (random) are randomly sampled by the method described in Sec.~\ref{subsubsec:wb_numerical}. Non-diminishing coherence influx $\mathbf{b}$ under internal dynamics $W$---Namely finite $W\mathbf{b}$ induces non-unit spectral radius $\rho(W) < 1$.}
	\label{fig:wb-rho}
\end{figure}
We found that $\|W\mathbf{b}\| > 0$ enforces $\lambda_\mathrm{max}(W) \neq 1$, which likely induces nonstationary ESP in numerical simulations. In addition, $\rho(W)$ and $|\mathrm{Re}(\lambda_{\mathrm{max}}(W))|$ are likely upper bounded by some function of $\|W\mathbf{b}\|/\|\mathbf{b}\|$ in our results. We employed QRCs governed by two types of Hamiltonians that are randomly generated to depict that relationship in Fig.~\ref{fig:sigma_vs_rho_sk}. One type of Hamiltonian uses the SK Hamiltonian, which is identical to the setup in Sec.~\ref{subsec:esp}, and the other type is generated by Hamiltonian generated by first sampling a uniform random matrix $\hat{O}_{\mathrm{rand}} \in \mathbb{R}^{4^N \times 4^N}$ and then taking $\hat{O}_{\mathrm{rand}}$ as a PTM and projecting it onto a completely positive trace-preserving (CPTP) manifold through a Choi form \cite{knee2018qpt}. A Choi form $\mathcal{C}$ of a PTM $\hat{O}_{\mathrm{rand}}$ can be calculated as follows \cite{daniel_2015}:
\begin{equation}
	\mathcal{C} = \sum_{i,j=0}^{4^N-1} (\hat{O}_{\mathrm{rand}})_{i,j} P_j^T \otimes P_i,
\end{equation}
where $P_i$ is the $i$-th Pauli string. All QRC in this experiment uses the identical reset-input encoding used in Sec.~\ref{subsec:esp}.

Figure~\ref{fig:wb-rho} numerically shows that, at least in simple Hamiltonian cases, $\|W\mathbf{b}\| > 0$ always ensures $\rho(\mathbb{P}_{\mathcal{Q}_b}W) < 1$ and $|\mathrm{Re}(\lambda_\mathrm{max}(W))| < 1$. This, in conjunction with the results in Fig.~\ref{fig:sigma_vs_rho_sk} and the second statement of Prop.~\ref{prop:suffice_convergence}, indicates a possibility of subspace nonstationary ESP under finite $W\mathbf{b}$.

\subsection{Multiplicative RC (mRC)}
\subsubsection{Theoretical results}
\begin{figure}[h]
\centering
\subfloat[$C_{tot}^{MC}$]{
\includegraphics[width=0.9\hsize]{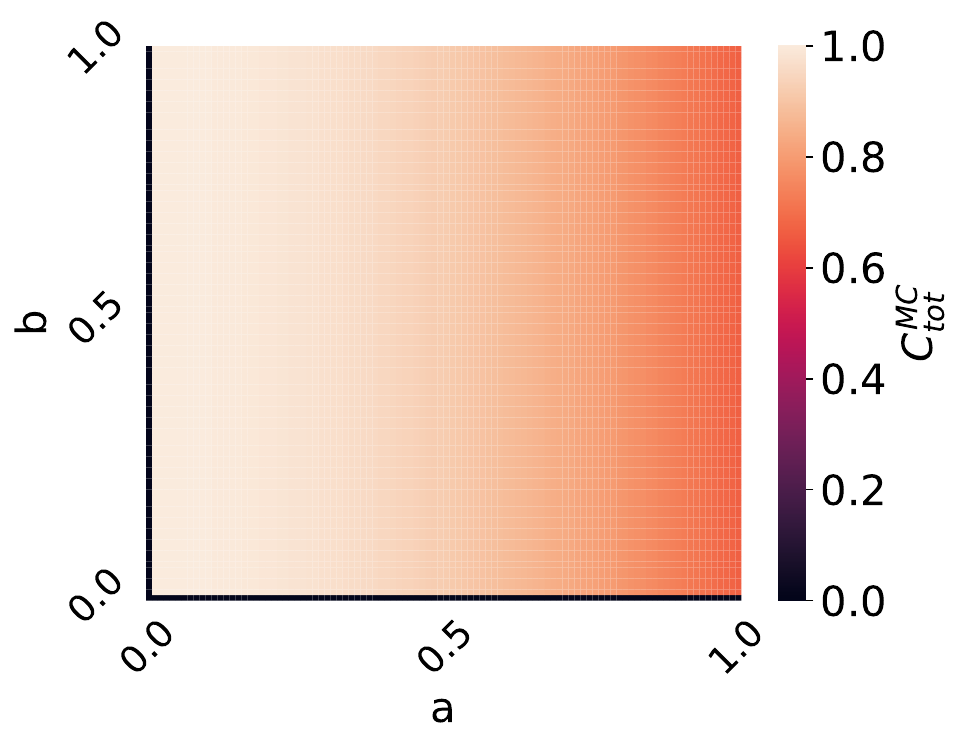}
}
\par
\subfloat[$C_{2+}^{IPC}$]{
\includegraphics[width=0.9\hsize]{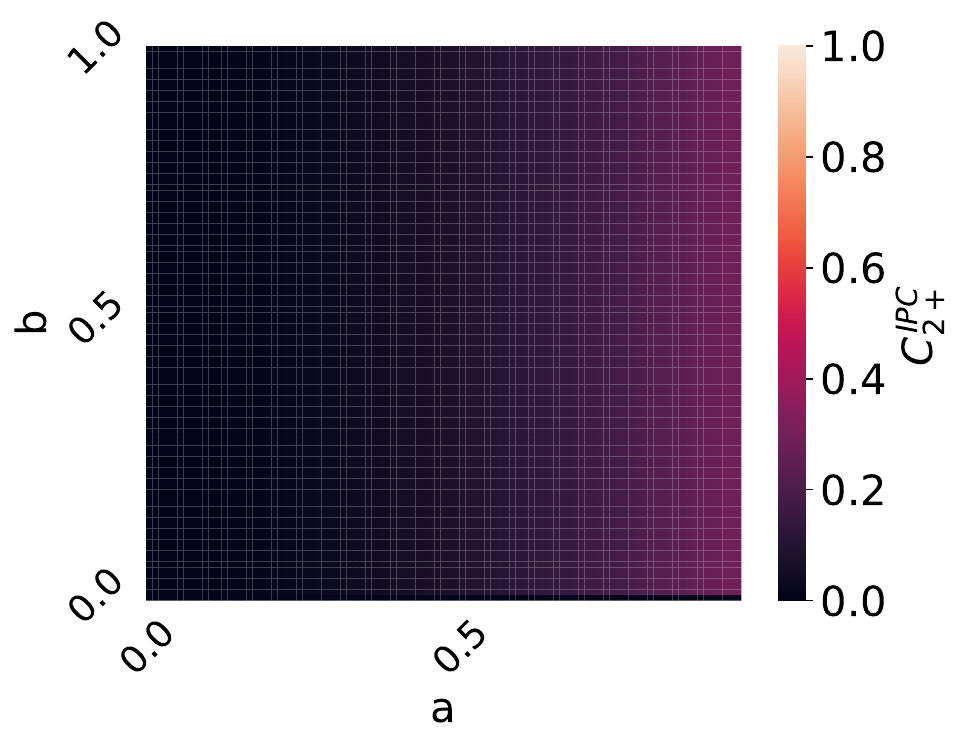}
}

\caption{a) $C_{tot}^{MC}$ and b) $C_{2+}^{IPC}\equiv \sum_{d\geq 2} C_{d}^{IPC}$ of mRC. It should be noted that negating $a$ is equivalent to negating $x$. Therefore, the results must be symmetric to a change of the sign of $a$.}
\label{fig:mrc-mc-ipc}
\end{figure}
The QR model governed by Eq.~\eqref{eqn:qrc_state_update} can, in its essence, be simplified to the RC model below with multiplicative inputs defined as follows:
\begin{dfn}{Multiplicative input reservoir computer}

Let a set of input $\mathcal{X} \subset \mathbb{R}$, and parameters $a, b \in \mathbb{R}$. A one-dimensional mRC has the following state update rule:
\begin{equation}
\label{eqn:mrc_state_update}
	x_{t+1} = a u_t x_t + b.
\end{equation}

If $|b| > 0$, we can divide both sides by $b$ and replace $\frac{x_t}{b}$ with $y_t$ to find the following equivalent state update rule:
\begin{equation}
\label{eqn:mrc-b-equivalence}
	y_{t+1} = a u_t y_t + 1.
\end{equation}
\end{dfn}

Here, we denote it as mRC and investigate its parameter dependency and non-linearity of the memory profile. While leave the comprehensive matrix analysis of the dynamics of QRC described in Eq.~\eqref{eqn:qrc_state_update} to future work because of its complexity.

Eq.~\eqref{eqn:mrc_state_update} corresponds to Eq.~\eqref{eqn:qrc_state_update} in the sense that
\begin{itemize}
\item $a$ in Eq.~\eqref{eqn:mrc_state_update} corresponds to $W$ in Eq.~\eqref{eqn:qrc_state_update}, and $|a|$ corresponds to $\rho(W)$.
\item $u_t$ in Eq.~\eqref{eqn:mrc_state_update} corresponds to $R(\mathbf{u}_t)$ in Eq.~\eqref{eqn:qrc_state_update}.
\item $b$ in Eq.~\eqref{eqn:mrc_state_update} corresponds to $\mathbf{b}$ in Eq.~\eqref{eqn:qrc_state_update}.

\end{itemize}

We can directly write the analytical form of the general state as follows:
\begin{equation}
\label{eqn:mrc-general}
\begin{aligned}
	x_{t+1} &= ax_tu_t + b\\
	&= a(ax_{t-1}u_{t-1} + b)u_t + b\\
	&= 	a^2 x_{t-1}u_{t-1}u_t + abu_t + b\\
	&\ \ \ \vdots\\
	&= a^{t+1}x_0\prod_{\tau=0}^{t}u_\tau + b\sum_{\tau=2}^{t}a^{\tau}\prod_{\sigma=0}^{\tau-1} u_{t-\sigma} + abu_t + b.
\end{aligned}
\end{equation}

If $|a| < 1$ and $t \gg 1$, 
\begin{equation}
\label{eqn:mrc-general2}
	x_{t+1} \sim b\sum_{\tau=2}^{t}a^{\tau}\prod_{\sigma=0}^{\tau-1} u_{t-\sigma} + abu_t + b.
\end{equation}

We have the following theorem:
\begin{thm} {Nonstationary ESP of mRC}
\label{thm:mrc-ns-esp}

An mRC has nonstationary ESP if and only if $0 < |a| < 1$ and $|b| > 0$.
\end{thm}

This theorem is a simplified version of Thm.~\ref{thm:sufficient_ns_esp} in which $\sigma_{\mathrm{max}}(W)$ is replaced by $|a|$. In addition, the injectiveness of $(I - WR(\mathbf{u}_t))\mathbf{b}$ is replaced by the injectiveness of $(1 - au_t)b$, which is satisfied if $|b| > 0$. Therefore, $|b| > 0$ is the simplified version of the injectiveness condition.

\begin{cor}{Result of boundedness in mRC}
\label{cor:mrc-bounded}

For an mRC, if the reservoir state space $\mathcal{S}$ is bounded, then $|b| > 0$ implies $|a| < 1$, which induces the nonstationary ESP of this mRC.
\end{cor}

Note that, based on the probability distribution from which $\{u_t\}$ are sampled, $|a| > 1$ cases can also have fading memory. This is because the multiplication of the input sequence helps in the convergence of the state difference on average, even if $|a| > 1$. For instance, if $\{u_t\}$ is drawn from a uniform distribution of $[-1, 1]$, because $\mathbb{E}(u_t^2) = \frac{1}{3}$, $0 < |a| < \sqrt{3}$ should satisfy nonstationary ESP. We have the following analytical form of memory capacity in mRC:

\begin{figure}[t]
\centering
\includegraphics[width=0.8\hsize]{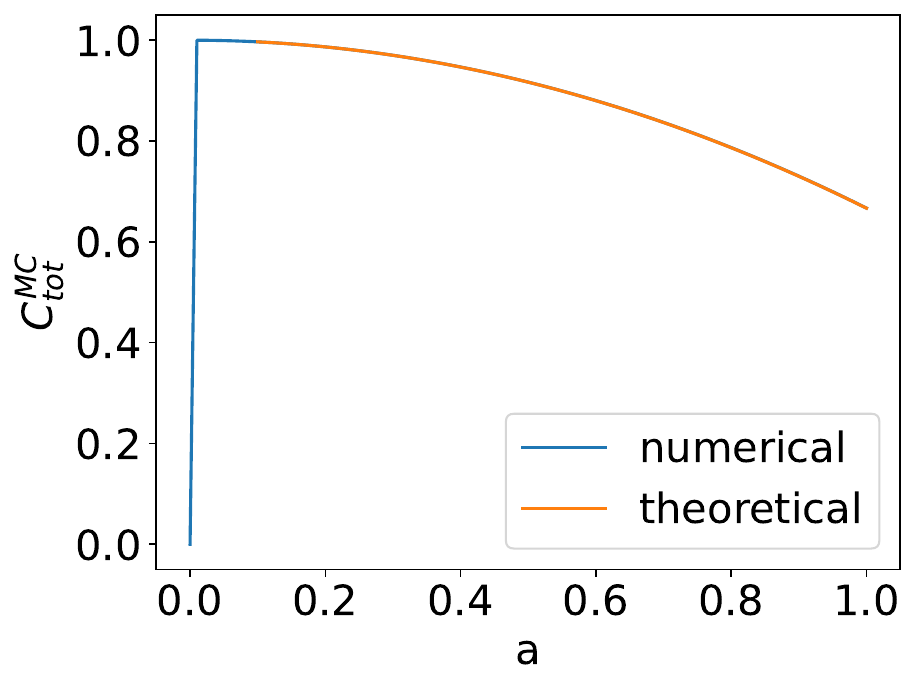}
\caption{Theoretical and numerical $C_{tot}^{MC}$ of mRC. The numerical calculation exactly matches the theoretical result of Eq.~\eqref{eqn:mrc_mc}. The theoretical curve is drawn in $[0.1,1.0]$ instead of $[0, 1.0]$ to make the overlap visible. The numerical calculation of $C_{tot}^{MC}$ has been carried out for $a \in \{k/99 \mid k \in \{0,1,\dots,99\}\}$ and $b=0.5$.
}
\label{fig:mrc-mc-theory}
\end{figure}
\begin{rem}{MC of mRC}

\label{rem:mrc-ipc}
If nonstationary ESP holds for an mRC under uniform-random input sequences---namely, $0 < |a| <  \sqrt{3}$---then,
\begin{equation}
\label{eqn:mrc_mc}
\begin{aligned}
	C_{tot}^{MC} &= C_{0}^{MC} = 1 - \frac{a^2}{3},
\end{aligned}
\end{equation}
\end{rem}
The reason why $C_{tot}^{MC} = C_{0}^{MC}$ is that every term containing $u_{t-k}$ for $k\geq 1$ is multiplied by $u_{t-k'}$ of $k'\neq k$, as shown in Eq.~\eqref{eqn:mrc-general2}. The factor $\frac{1}{3}$ appears as a result of taking uniform distribution for our inputs $\{u_t\}$. Please refer to the appendix for the proof.

\subsubsection{Numerical results}
We have numerically examined the MC and IPC of mRC. Fig.~\ref{fig:mrc-mc-ipc} shows the MC and IPC of the mRC with different parameter configurations. We can observe that MC becomes larger as $a$ increases while IPC decreases, meaning that a larger $a$ makes the system more linear. The invariance of the information processing capability by changing $b$, which is implied by Eq.~\eqref{eqn:mrc-b-equivalence}, can also be observed. In Fig.~\ref{fig:mrc-mc-theory}, we compare the theoretical MC in Eq.~\eqref{eqn:mrc_mc} and the numerical experimental result of $b =0.5$. The numerical relationship of MC and $b$ exactly matches the theoretical result.

As shown in Eq.~\eqref{eqn:mrc-general}, all degree $d$ non-linear terms have the form of $\prod_{\tau=t}^{t-d-1}u_\tau$. Those types of non-linearities also present in the output signals of QRC, and we expect that those RCs are appropriate for temporal tasks that have such types of non-linearities.

\subsection{Additional numerical experiments}
\label{subsec:num_exp}
\subsubsection{Reproduction of the dynamical phase transition \cite{Mart_nez_Pe_a_2021}}
\label{subsubsec:dymical_phase}
\begin{figure}[t]
\centering

\subfloat[NS ESP]{
\label{subfig:sk_esp_vanilla}
	\includegraphics[width=0.8\hsize]{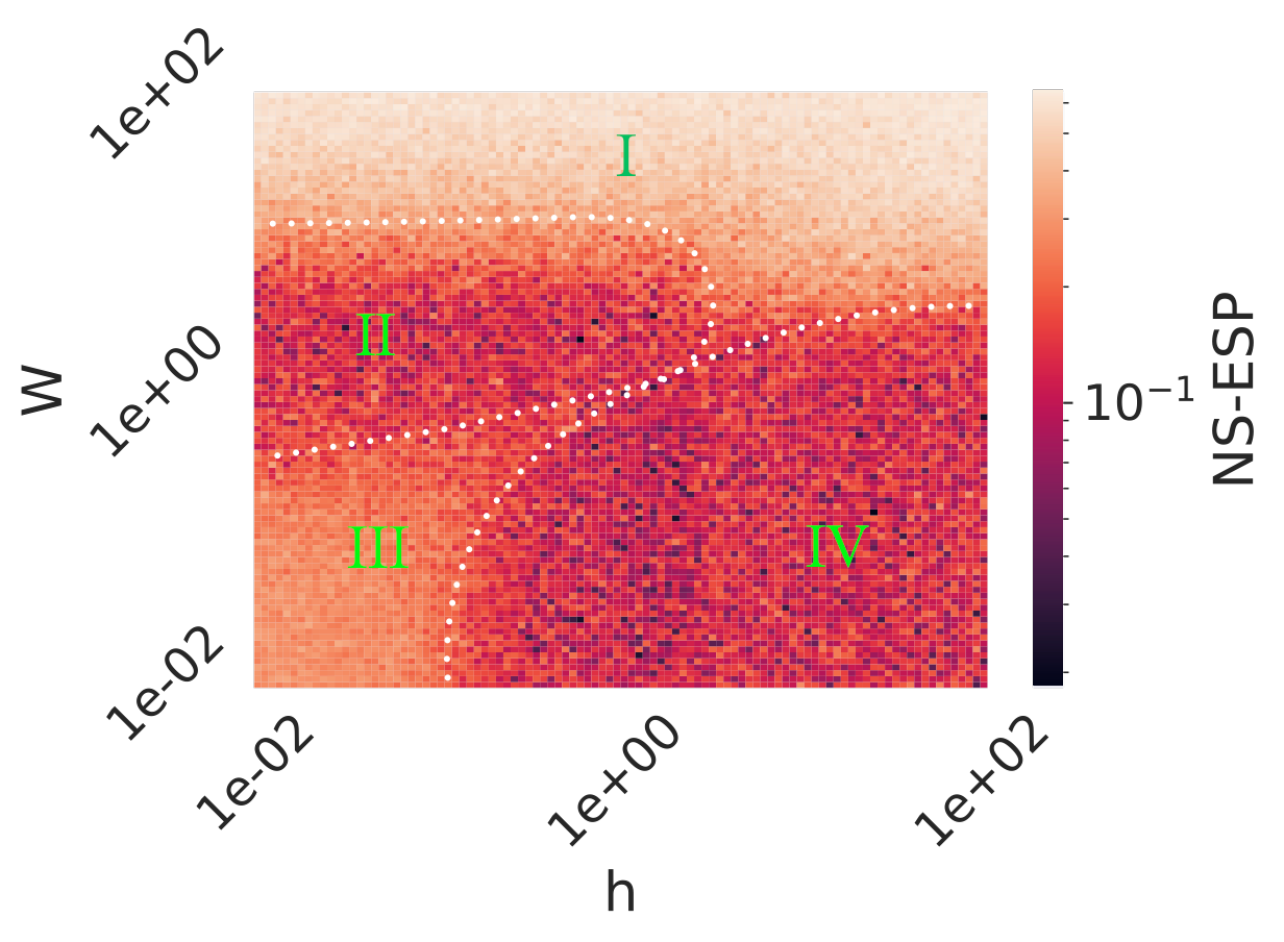}
}
\par
\subfloat[$\rho(W)$]{
\label{subfig:sk_rho_grid}
	\includegraphics[width=0.8\hsize]{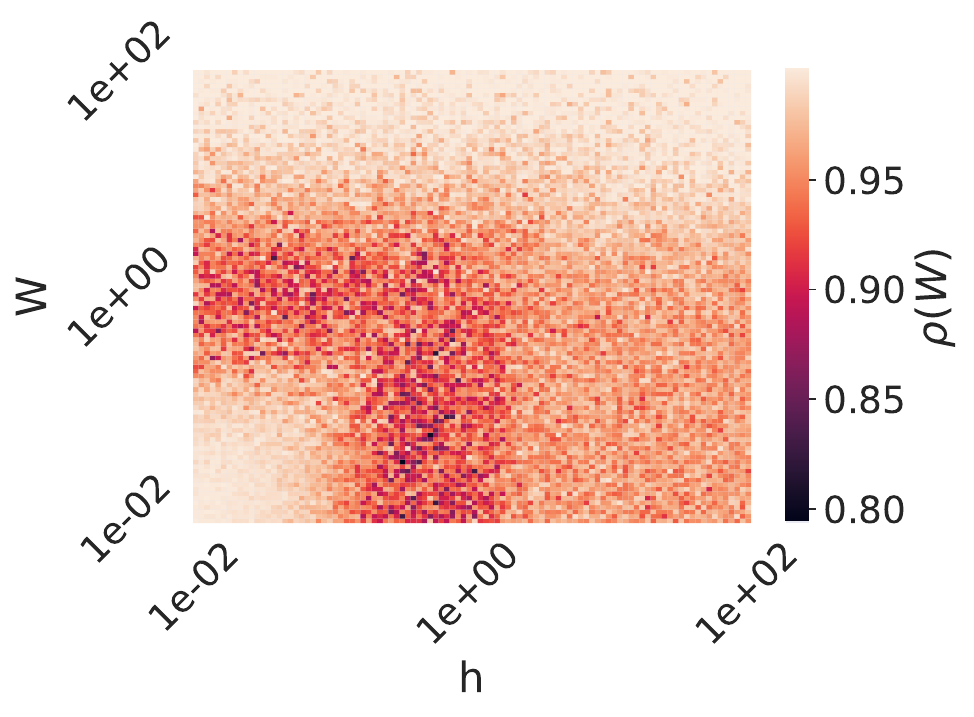}
}

\caption{Numerical simulation analysis of a 2-qubit system whose PTM composed of SK Hamiltonians and single qubit reset operations that is prescribed in Eq.~\eqref{eqn:sk_hamiltonian} and \eqref{eqn:reset_unitary_input_sk}. (a) Nonstationary ESP indicator: the left-hand side of Eq.~\eqref{eqn:esp_ext2} computed by random inputs $\{\mathbf{u}_t \sim \mathrm{Uniform}([-1, 1])\}$ such that $|\{\mathbf{u}_t\}|=200$ and $w=10$. (b) Spectral radius distribution. Letters I--IV in (a) indicate the phase diagram discussed in existing work \cite{Mart_nez_Pe_a_2021}. White dotted lines in (a) were added by this authors for visual aids to distinguish each region.
 }
  \label{fig:sk}
\end{figure}
We conducted additional numerical experiments using the same setup as in Sec.~\ref{subsec:esp}. Here, we tried to replicate the dynamical phase transition \cite{Mart_nez_Pe_a_2021} intrinsic to this system by the spectral radius of PTMs.  

We computed the spectral radius $\rho(W)$ of PTMs and the nonstationary ESP indicator (Eq.~\ref{eqn:ns_esp_indicator}) for randomly sampled SK Hamiltonians (Eq.~\eqref{eqn:sk_hamiltonian}) for every configuration in $J_s, K \in \{10^{-2 + 4k/99} \mid k\in \{0,1,\dots,99\}\}$. Here, $w=20$ and $\|\{\mathbf{u}_t\}\| = 200$ were used for the nonstationary ESP indicator calculations.

The results are in Fig.~\ref{fig:sk}. Figure~\ref{subfig:sk_esp_vanilla} corresponds to the dynamical phase transition of QRC that has been explored in existing work \cite{Mart_nez_Pe_a_2021}. Each region of different fading memory characteristics is labeled as I to IV, as in the existing work \cite{Mart_nez_Pe_a_2021}. We can see that the nonstationary ESP indicator results in Fig.~\ref{subfig:sk_esp_vanilla} have similar phase structure corresponding to the spectral radius results in Fig.~\ref{subfig:sk_rho_grid}, which, in part, supports our theorem in which the spectral radius $\rho(W)$ takes an important role in the system's ESP.

\begin{figure}
\centering

\subfloat[NS ESP]{
\label{subfig:sk_esp_vanilla_disc}
	\includegraphics[width=0.7\hsize]{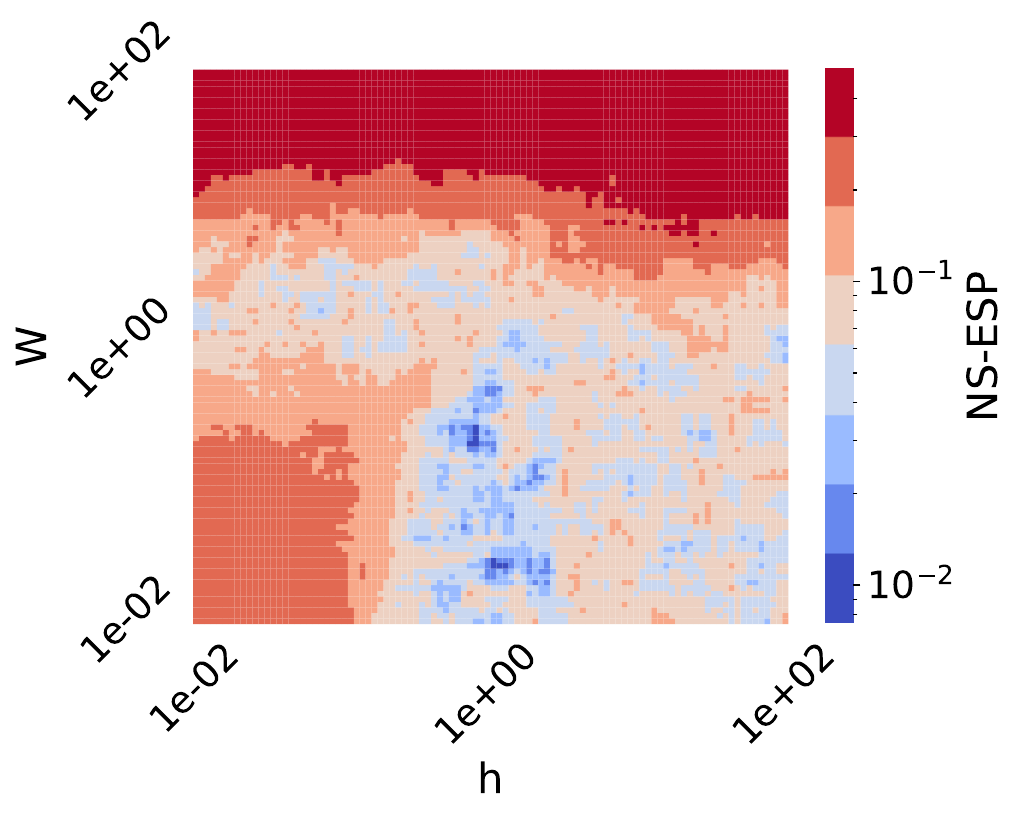}
}
\par
\subfloat[$\rho(W)$]{
\label{subfig:sk_rho_grid_disc}
	\includegraphics[width=0.7\hsize]{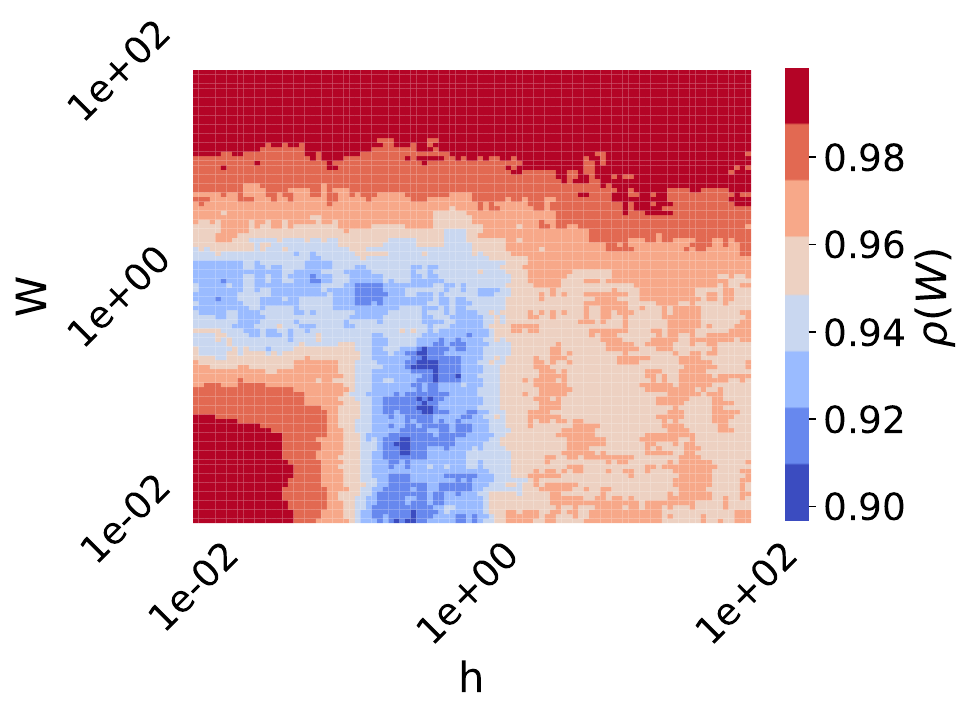}
}
\par
\subfloat[$\rho_{\mathrm{eff}}(W)$]{
\label{subfig:sk_rho_eff_grid_disc}
	\includegraphics[width=0.7\hsize]{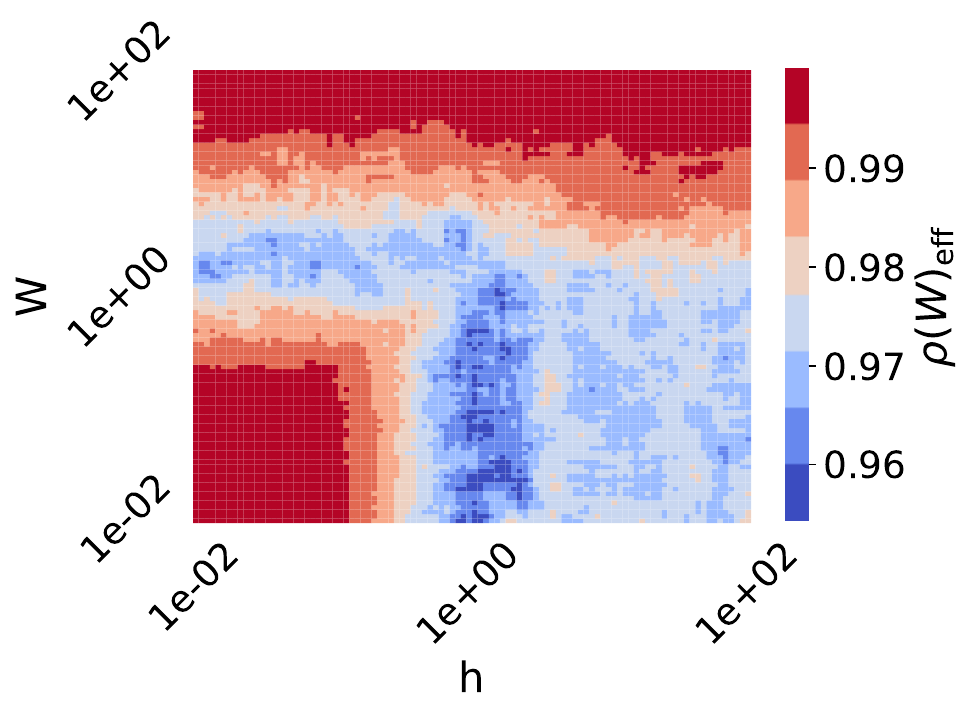}
}
\par
\subfloat[$\frac{\rho_{\mathrm{eff}}(W)}{\rho(W)}$]{
\label{subfig:sk_rho_eff_over_rho_grid_disc}
	\includegraphics[width=0.7\hsize]{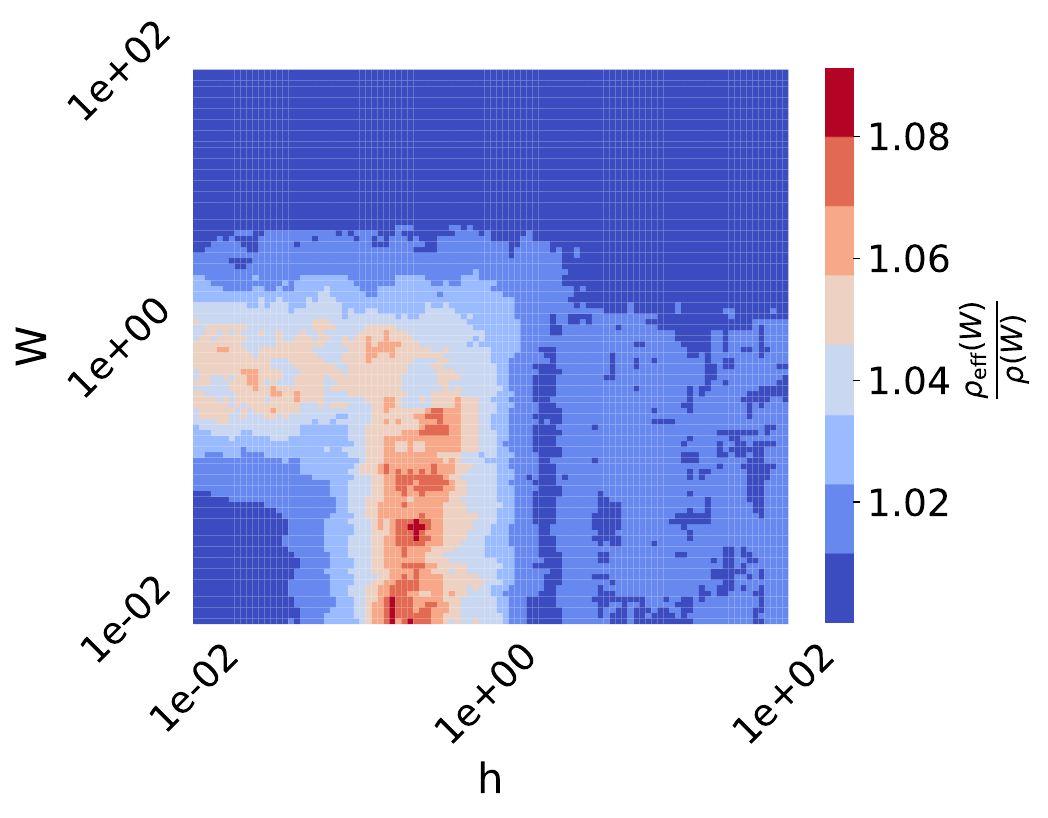}
}
\caption{
The discretized heatmap of convolved and geometrically averaged (a) the nonstationary ESP indicator, (b) $\rho(W)$, (c) $\rho_\mathrm{eff}(W)$ and (d) $\rho_\mathrm{eff}(W) / \rho(W)$.
}
  \label{fig:sk_disc}
\end{figure}
To further analyze the relationship between the nonstationary ESP and the spectral radius of PTM, the results of Fig.~\ref{fig:sk} was convolved and geometrically averaged by a kernel $K$ of the following form:
\begin{equation}
	K = \begin{pmatrix}
		1 & 1 & 1 & 1 & 1\\
		1 & 1 & 1 & 1 & 1\\
		1 & 1 & 0 & 1 & 1\\
		1 & 1 & 1 & 1 & 1\\
		1 & 1 & 1 & 1 & 1\\
	\end{pmatrix},
\end{equation}
and depicted by discretized heatmaps in Fig.~\ref{fig:sk_disc}. The results showed that, although $\rho(W)$  (Fig.~\ref{subfig:sk_rho_grid_disc}) and $\rho_\mathrm{eff}(W)$ (Fig.~\ref{subfig:sk_rho_eff_grid_disc}) have globally similar phase structures to the nonstationary ESP indicator (Fig.~\ref{subfig:sk_esp_vanilla_disc}), the fading memory of the parameter region indexed as IV was rather underestimated by $\rho(W)$. We argue that this was caused by the discrepancy between $\rho(W)$ and $\rho_\mathrm{eff}(W)$ in the region II, as shown in Fig.~\ref{subfig:sk_rho_eff_over_rho_grid_disc}, where $\rho_\mathrm{eff}(W)$, which represents fading memory strength under input-driven dynamics, can be underestimated by $\rho(W)$.

\begin{figure}[t]
\centering
\subfloat[$C_{tot}^{MC}$]{
\label{subfig:sk_mc_tot_grid}
\includegraphics[width=0.8\hsize]{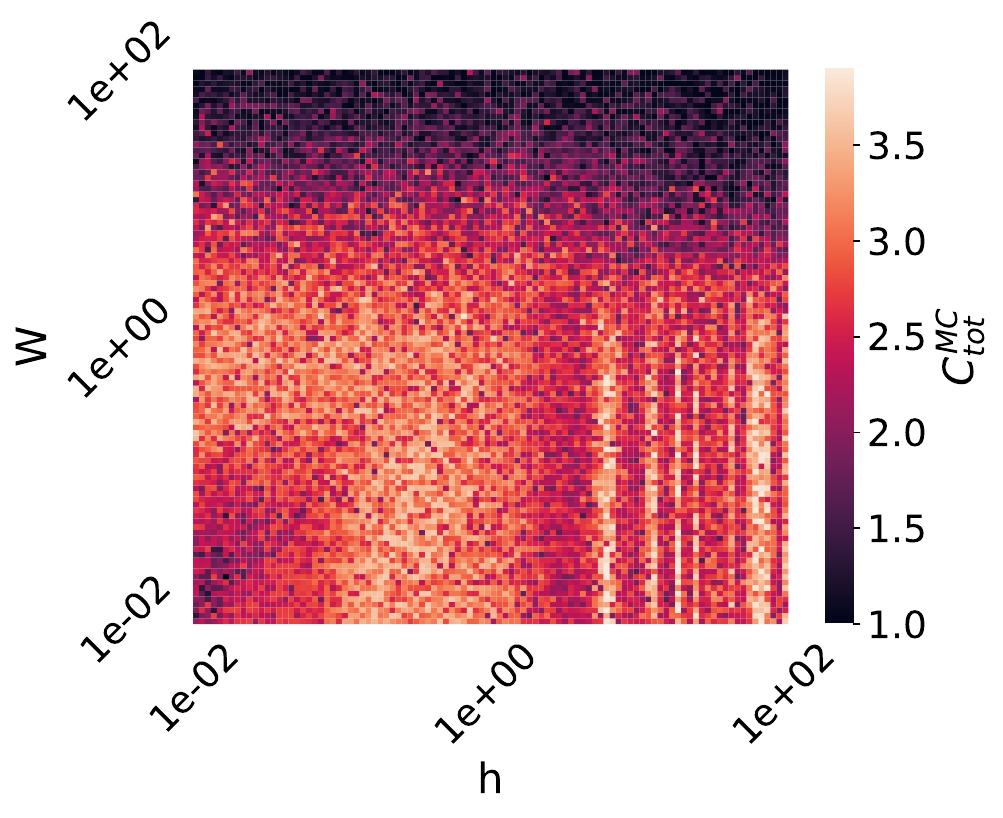}
}
\par
\hspace{0.3cm}\subfloat[Variance decay]{
\label{subfig:sk_var_decay}
\includegraphics[width=0.82\hsize]{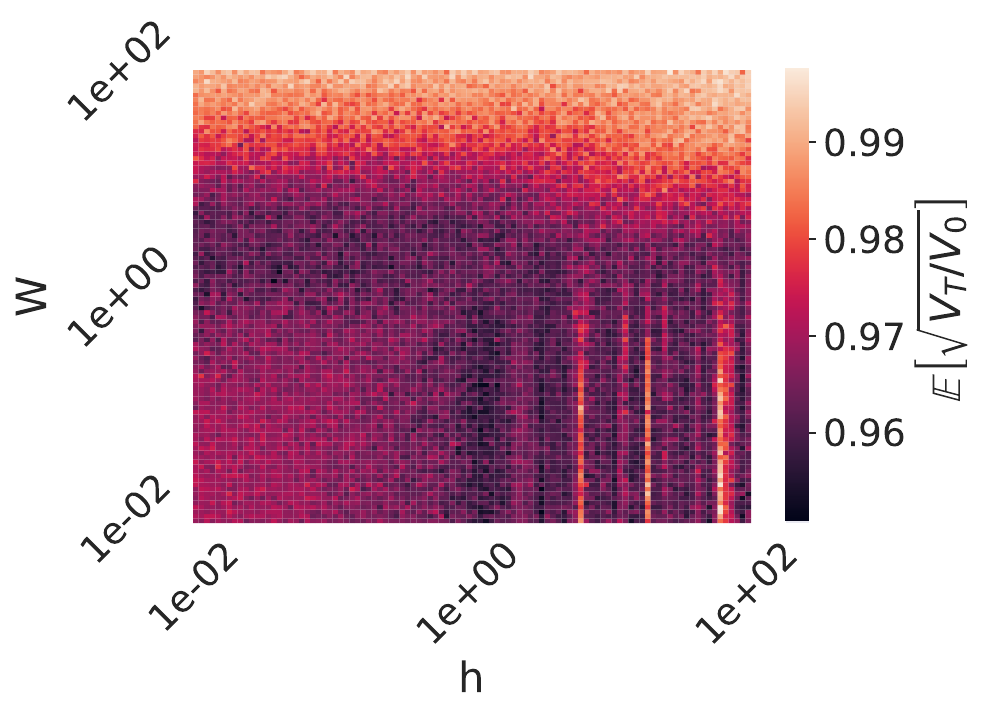}
}
\caption{
MC calculation results for the SK model. Except for the lower left, the low MC region corresponds to weak contractivity. (a) $C_{tot}^{MC}$, (b) Average of the variance decay calculated by Eq.~\eqref{eqn:var_decay}.
}
\label{fig:sk_mc_var_decay}
\end{figure}

\begin{figure}[t]
\centering

\subfloat[]{
\label{subfig:sk_mc1_rho}
	\includegraphics[width=0.8\hsize]{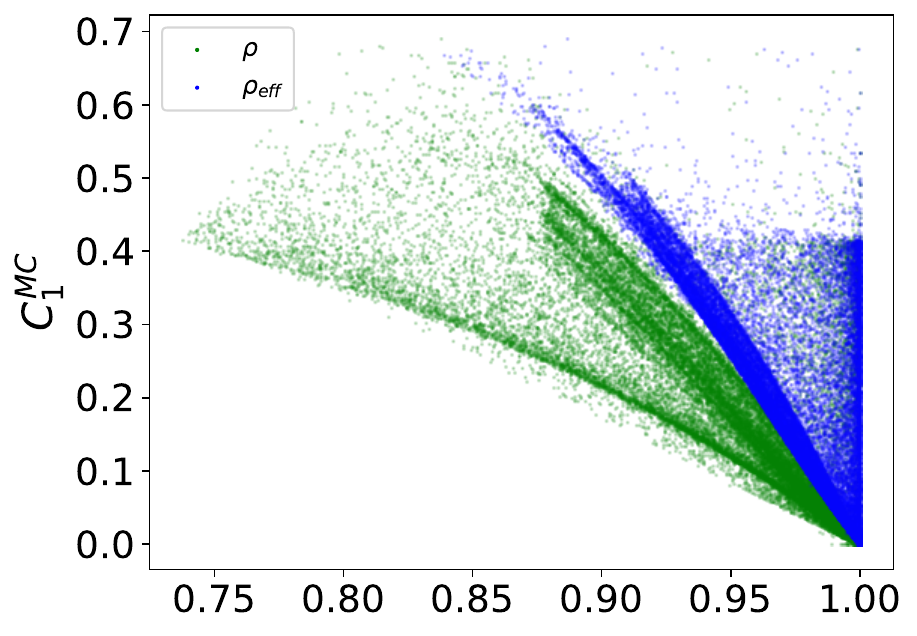}
}
\par
\subfloat[]{
\label{subfig:sk_mc_tot_rho}
	\includegraphics[width=0.8\hsize]{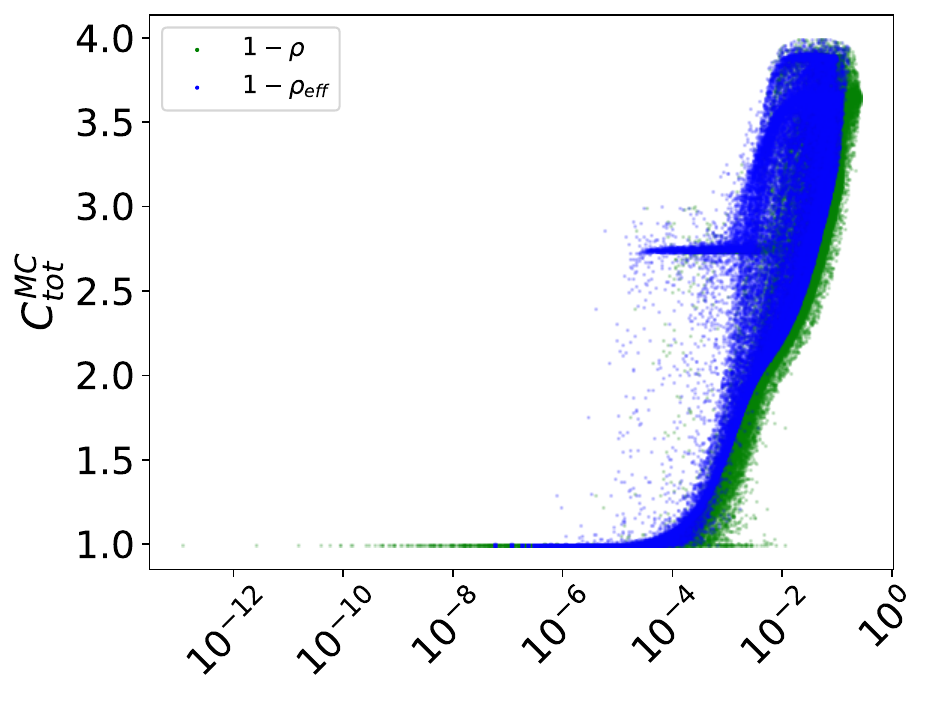}
}
\caption{(a) $C_{1}^{MC}$ by $\rho(W)$ (green) and $\rho_{\mathrm{eff}}(W)$ (blue). (b) $C_{tot}^{MC}$ by $1 - \rho(W)$ (green) and $1 - \rho_{\mathrm{eff}}(W)$ (blue). Each dot corresponds to different parameter configurations of the Hamiltonian in Eq.~\eqref{eqn:sk_hamiltonian}.
}
\label{fig:sk_mc}
\end{figure}

\subsubsection{Memory profile}
In addition, we numerically calculated the MC of QRCs with SK Hamiltonians for all of the configurations in Sec.~\ref{subsubsec:dymical_phase} to identify if there were any relationships between the nonstationary ESP and MC. First, we compared $C_{tot}^{MC}$ of each QRC with the variance decay of the following form:
\begin{equation}
\label{eqn:var_decay}
	\mathrm{Var}_{decay}(w, \{\mathbf{u}_t\}) \equiv \frac{\sqrt{\min\left[\overline{\mathrm{Var}}_w^t(f; s_0), \overline{\mathrm{Var}}_w^t(f; s'_0)\right]}}{\sqrt{\min\left[\overline{\mathrm{Var}}_w^w(f; s_0), \overline{\mathrm{Var}}_w^w(f; s'_0)\right]}}.
\end{equation}
Here, the numerator of the right-hand side of Eq.~\eqref{eqn:var_decay} quantifies windowed variance at time $t$, and the denominator quantifies windowed variance at time $w$. Both calculations use a time window of size $w$, and the variance decay, in total, describes how stationary the output signals are, assuming the mean of the output signals does not change.

This term was an important addition to the traditional ESP in the nonstationary ESP definition. We computed the variance decay using the same setup as in Sec.~\ref{subsubsec:dymical_phase}, the MCs were calculated using an input sequence of length $1.3 \times 10^6$, each of which was independently sampled from a uniform distribution within range $[-1, 1]$, where $3 \times 10^5$ outputs were discarded as washouts and the remaining $10^6$ outputs were used in the calculations. The results are depicted in Fig.~\ref{fig:sk_mc_var_decay}. There is a region where strong nonstationary ESP does not ensure relatively high $C_{tot}^{MC}$ compared with other surrounding regions, as shown in the lower right part of Fig.~\ref{subfig:sk_mc_tot_grid}. However, the high-variance-decay regions in the lower-right part of the Fig.~\ref{subfig:sk_var_decay} correspond to the high $C_{tot}^{MC}$ in Fig.~\ref{subfig:sk_mc_tot_grid}. We hypothesize that variance decay near one ensures stable dynamics and induces a stronger fading memory, provided it has nonstationary ESP.

In addition, Fig.~\ref{fig:sk_mc} shows that there exists an almost linear relationship between the lower bound of $C_{1}^{MC}$ and $\rho$/$\rho_{\mathrm{eff}}$, and an almost log-linear relationship between $C_{tot}^{MC}$ and $(1 - \rho)$/$(1-\rho_{\mathrm{eff}})$. These results are similar to those for mRC in Fig.~\ref{fig:mrc-mc-theory} in a sense that $\rho$, which is a corresponding parameter to $a$ in mRC, has a direct relationship with $C_{1}^{MC}$ or $C_{tot}^{MC}$, as in Rem.~\ref{rem:mrc-ipc}.

\section{Conclusions}

In the present paper, we theoretically analyzed a condition of a dissipative quantum system to satisfy nonstationary ESP and showed that it is related to the coherent interaction with its environment, and the spectral radius of the internal dynamics in PTM form. To make the analysis more straightforward for practical implementation, we divided the system dynamics into the internal dynamics part and the input encoding part. The results indicate that conditions necessary for the traditional ESP of QRC is similar to the conditions necessary for the traditional ESP of classical ESN in a sense that the spectral radius and Schur stability of PTM and recurrent weight are important, respectively for QRC and ESN. In addition, for the nonstationary ESP, QRC needs to have consistent input encoding method to ensure variable outputs under variable inputs. To demonstrate these theoretical results, we numerically showed that the nonstationary ESP indicators had an almost log-linear relationship with the effective input-driven spectral radius of PTM, using the SK Hamiltonian and reset-input encoding method. This suggests that the spectral property of PTM has huge importance  in the ESP of QRCs.

Because it is complex to analyze PTM in detail, we devised a simple one-dimensional RC model, mRC, and showed similar results to QRC cases. That is, model parameters of mRC, that are similar to spectral radius and coherence influx in QRC, determine the nonstationary ESP of mRC.  Specifically, the multiplicative factor $a$, which is similar to the spectral radius in QRC, determines the length of memory in mRC, and the additive factor $b$, which is similar to the coherence influx in QRC, ensures finite output signals under any inputs. Furthermore, the analytical form of $C_{tot}^{MC}$ using those parameters was provided for mRC.

In addition, motivated by the above results, which imply the spectral radius of PTM being important for the fading memory property of QRC, we replicated dynamical phase transition results \cite{Mart_nez_Pe_a_2021} by the spectral radius of PTM, using QRC governed by SK Hamiltonian and reset-input encoding. The results confirmed that the spectral radius of PTM can explain the fading memory of the QRC. Furthermore, we found that the lower bound of $C_{1}^{MC}$ and $C_{tot}^{MC}$ of QRC governed by SK Hamiltonian and reset-input encoding has an almost linear/log-linear relationship with the spectral radius, respectively. This suggests that there will be a possible analytical relationship between the linear memory capacity and the spectral radius of PTM in QRC, similar to the analytical relationship between the linear memory capacity and the parameter $a$ in mRC.

Our work provides a methodology for analyzing QRC using PTM and an experimentally checkable condition of a QRC having fading memory. Specifically, we showed that the ESP of a QRM is heavily dependent on the spectral radius of PTM describing the non-input-driven dynamics, the interaction between the system and its quantum coherent environment, and the input encoding method relative to the non-input-driven dynamics. This work will lead to a better understanding of QRC and other information processing techniques that exploit open quantum systems.

\section{Acknowledgements}
This work was supported by the MEXT Quantum Leap Flagship Program (MEXT Q-LEAP) grant no. JPMXS0120319794.

\begin{appendix}

\section{Pauli transfer matrix formulation}

\begin{dfn} {Multi-qubit Bloch vector in Pauli transfer matrix formulation (the coherence vector) }\\
\label{dfn:multi-qubit-bloch}
A set of all $N$-qubit physical state $\mathbf{r}(\rho)$ and $|\rho\rangle\rangle$ in a Pauli transfer matrix formulation, $\mathcal{Q}(N)$ and $\hat{\mathcal{Q}}(N)$, are respectively defined as 
\begin{equation}
\label{eqn:q_n}
\begin{aligned}
\mathcal{Q}(N) &= \bigg\{\mathbf{r}(\rho) \in \mathbb{R}^{4^N-1}\ \text{s.t}\ \mathbf{r}(\rho)_i = \mathrm{tr}\left(P_{i}\rho\right)  \bigg|\\
&i \geq 1,\ \rho \in \mathbb{C}^{2^N \times 2^N}\ \text{s.t}\ \rho = \rho^\dagger,\ \mathrm{tr}(\rho) = 1\ \mathrm{and}\ \rho \succeq 0\bigg\},\\
\hat{\mathcal{Q}}(N) &= \bigg\{|\rho\rangle\rangle \equiv \begin{pmatrix}
	1\\
	\mathbf{r}(\rho)
\end{pmatrix} \bigg| \mathbf{r}(\rho) \in \mathcal{Q}(N)\bigg\},
\end{aligned}
\end{equation}
\end{dfn}
$\mathbf{r}(\rho)$ above is called a multi-qubit Bloch vector, or a coherence vector. For this manuscript, $\mathbf{r}(\rho)$ is written as $\mathbf{r}$ for simplicity. In addition, $|\rho\rangle\rangle^T$ is written as $\langle\langle\rho|$ because of the analogy to the bracket notation of quantum mechanics.

\begin{lem} {System dynamics under Pauli transfer matrix}
\label{lem:state-update-pauli}

Let an $N$-qubit quantum channel be $\mathcal{E}$ and the corresponding Pauli transfer matrix be $\hat{O}_{\mathcal{E}}: \hat{\mathcal{Q}}(N) \to \hat{\mathcal{Q}}(N)$; then,
	\begin{enumerate}
	\item $\hat{O}_{\mathcal{E}}$ can be written as $\hat{O}_{\mathcal{E}} = \begin{pmatrix}
	1 & \mathbf{0}^T\\
	\mathbf{b} & W
\end{pmatrix}$, where $\mathbf{b} \in \mathbb{R}^{4^N-1}$ and $W \in \mathbb{R}^{(4^N -1) \times (4^N -1)}$.
	\item Given an initial state $|\rho^{(0)}\rangle\rangle \equiv \begin{pmatrix}
	1 \\
	\mathbf{r}^{(0)} \end{pmatrix} \in \hat{\mathcal{Q}}(N)$, the channel application can be written as 
	\begin{equation}
		|\rho^{(n+1)}\rangle\rangle = \hat{O}_{\mathcal{E}}|\rho^{(n)}\rangle\rangle,
	\end{equation}
	where 
	\begin{equation}
	|\rho^{(n)}\rangle\rangle = \begin{pmatrix}
	1\\
		\mathbf{r}^{(n)}
\end{pmatrix} = \begin{pmatrix}
	1 \\
	\sum_{i=0}^{n-1}W^i\mathbf{b} + W^{n}\mathbf{r}^{(0)}
	
\end{pmatrix}.
	\end{equation}
\end{enumerate}
\end{lem}

\begin{proof}
	1.  $\mathrm{tr}\bigg(I \sum_k K_k P_j K_k^\dagger \bigg) = \mathrm{tr}\bigg(\sum_k K_k^\dagger I K_k P_j \bigg) = \mathrm{tr}\bigg(P_j\bigg) = \delta_{j,0}$. Therefore, the first row, except for the first column, is all zero.

	2. Given $\{K_k\}$ as a Kraus representation of a quantum channel $\mathcal{E}$ is,
	\begin{equation}
		|\mathcal{E}(\rho)\rangle\rangle_i = \sum_k\mathrm{tr}\left(P_i K_k\rho K_k^\dagger\right).
	\end{equation}
	On the other hand, from Eq.~\eqref{eqn:transfer_matrix_kraus},
	
	\begin{equation}
		\hat{O}_\mathcal{E}|\rho\rangle\rangle_i = \sum_k\sum_j \mathrm{tr}\left(P_i K_k P_j K_k^\dagger\right)\mathrm{tr}\left(P_j \rho\right).
	\end{equation}
	From the first statement of Lem.~\ref{lem:mq_bloch}, we can replace $\sum_j P_j\mathrm{tr}(P_j\rho)$ with $\rho$ and obtain
	\begin{equation}
		\hat{O}_\mathcal{E}|\rho\rangle\rangle_i = \sum_k\mathrm{tr}\left(P_i K_k \rho K_k^\dagger\right) = |\mathcal{E}(\rho)\rangle\rangle_i.
	\end{equation}

\end{proof}

\begin{lem}{Property of a Pauli transfer matrix}
\label{lem:transfer_matrix}

An $N$-qubit Pauli transfer matrix
\begin{equation*}
	\hat{O}_{\mathcal{E}} = \begin{pmatrix}
	1 & \mathbf{0}^T\\
	\mathbf{b} & W
\end{pmatrix} : \hat{\mathcal{Q}}(N) \to \hat{\mathcal{Q}}(N)
\end{equation*}
 corresponding to a quantum channel ${\mathcal{E}}$ has the following properties:
\begin{enumerate}
	\item $\mathbf{b} \in \mathcal{Q}(N)$.

	\item $\rho(W) \leq 1$.
	\item $\rho(W) < 1 \Rightarrow \forall \mathbf{r},\ (\hat{O}_{\mathcal{E}})^n \mathbf{r} \underset{n \to \infty}\to (I - W)^{-1}\mathbf{b}$.
	\item For an eigenvalue of $W$: $\lambda$, if $|\lambda| = 1$, then $\mu_W(\lambda) = \gamma_W(\lambda)$.
	\item 
		$\mathbf{b} = \mathbf{0} \Leftrightarrow \mathcal{E} \text{ is unital}$.

\end{enumerate}
Here, $\rho(A) \equiv \max_i |\lambda_i(A)|$ is the spectral radius of the matrix $A$, and $\mu_A(\lambda)$ and $\gamma_A(\lambda)$ are the algebraic and geometric multiplicity of a matrix $A$ with respect to the eigenvalue $\lambda$.

\end{lem}
\begin{proof}
1. An application of $\hat{O}$ to a completely mixed state as $\hat{O}\begin{pmatrix}
	1\\
	\mathbf{0}
\end{pmatrix} = \begin{pmatrix}
	1\\
	\mathbf{b}
\end{pmatrix}$ must yield a physical state. Therefore,  $\mathbf{b} \in \mathcal{Q}(N)$.

2. If $\rho(W) > 1$, selecting a physical state $c_N\mathbf{r}^{(0)}$ that is non-orthogonal to $\vec{\mathbf{\lambda}}_{max}$, one of the eigenvectors corresponding to the maximum eigenvalue, reads $\|\mathbf{r}^{(n)}\| \underset{n\to \infty}\to \infty$, and it is always possible. Therefore, $\rho(W) \leq 1$ by contradiction.
 
3. 
Let $A \equiv \sum_{i=0}^{n-1}W^i\mathbf{b}$; then, $A - WA = (I-W)A = \mathbf{b} - W^n \mathbf{b} = (I - W^n)\mathbf{b}$. Therefore, $A = (I-W)^{-1}(I - W^n) \mathbf{b}$. If $\rho(W) < 1$, $I - W$ is nonsingular and $W^n\mathbf{b} \underset{n\to \infty}\to 0$. Therefore,
\begin{equation}
\begin{aligned}
	\mathbf{r}^{(n)} &= \left(I - W\right)^{-1}(I - W^n)\mathbf{b} + W^n\mathbf{r}^{(0)} \\
		&\underset{n \to \infty} \to \left(I - W\right)^{-1}\mathbf{b} \equiv \mathbf{r}_\infty.
\end{aligned}
\end{equation}

	4. If $\exists i\ \textrm{s.t.}\ \mu_R(\lambda_i) \neq \gamma_R(\lambda_i)$, then there exists a Jordan canonical form $W = MJM^{-1}$, having a Jordan block such as $J_{\lambda_i} =\begin{pmatrix}
		e^{i\theta} & 1 & \cdots\\
		0 & e^{i\theta} & \cdots\\
		0 & 0 & \ddots\\
	\end{pmatrix}$. $n$ times application of $\hat{O}$ leads to a term $J_{\lambda_i}^n = \begin{pmatrix}
				e^{in\theta} & ne^{i(n-1)\theta} & \cdots\\
		0 & e^{in\theta} & \cdots\\
		0 & 0 & \ddots\\
	\end{pmatrix} \underset{n \to \infty} \to \begin{pmatrix}
				\cdot & \infty & \cdots\\
		0 & \cdot & \cdots\\
		0 & 0 & \ddots\\
	\end{pmatrix}$. However, $\forall n,\ \|\mathbf{r}^{(n)}\| \leq c < \infty$, so $\forall i,\ \mu_R(\lambda_i) = \gamma_R(\lambda_i)$ by contradiction. 
	
	5.  First, $\mathbf{r}(\frac{I}{2^N}) = \mathbf{0}$, because $\forall i\geq 0,\ \mathrm{tr}(IP_i) = 0$. Therefore, if, $\mathbf{b}=\mathbf{0}$, then 
	\begin{equation}
	\label{eqn:unital}
		\hat{O}_{\mathcal{E}}\begin{pmatrix}
			1\\
			\mathbf{0}
		\end{pmatrix} = \begin{pmatrix}
			1\\
			\mathbf{b}
		\end{pmatrix} = \begin{pmatrix}
			1\\
			\mathbf{0}
		\end{pmatrix}.
	\end{equation} That is, $\mathcal{E}$ is unital.\\
	 Conversely, if $\mathcal{E}$ is unital, then $\mathbf{b} = 0$ is obvious from Eq.~\eqref{eqn:unital}. Especially when $\mathcal{E}$ is unitary, 
	let the Kraus operator of a unitary channel be $U$; then, $\mathrm{tr}\bigg(P_i U I U^\dagger \bigg) = \mathrm{tr}\bigg(P_i\bigg) = \delta_{i,0}$, so $\mathbf{b} = \mathbf{0}$. 

	\end{proof}

\begin{lem}{Property of a unitary channel in the Pauli transfer matrix formulation}
\label{lem:unitary_process}

Let a Pauli transfer matrix of a quantum channel $\mathcal{E}$ be $\hat{O}_{\mathcal{E}} = \begin{pmatrix}
	1 & \mathbf{0}^T\\
	\mathbf{b}_\mathcal{E} & R_{\mathcal{E}}
\end{pmatrix}$; then,
\begin{equation}
	\exists \mathcal{W}  \subset \mathrm{O}(4^N-1) \text{ s.t. } R_{\mathcal{E}} \in \mathcal{W} \Leftrightarrow \mathcal{E} \text{ is unitary.}
\end{equation}

\end{lem}

\begin{proof}
If for a given Pauli transfer matrix $\hat{O} = \begin{pmatrix}
	1 & \mathbf{0}^T\\
	\mathbf{b} & W
\end{pmatrix}$, $W \in \mathrm{O}(4^N-1)$, then $W: \mathcal{Q}(N) \to \mathcal{Q}(N)$ is a bijection. In addition, from the fifth statement of Lem.~\ref{lem:transfer_matrix}, $\mathbf{b} = 0$. Thus, there exists an inverse map $W^{-1} = W^T: \mathcal{Q}(N) \to \mathcal{Q}(N)$ s.t. $W \circ W^{-1} = W^{-1} \circ W = Id$. Then, $\forall \mathbf{r}\in \mathcal{Q}(N),\ W\mathbf{r} \in \mathcal{Q}(N)$, and $W^{-1}\mathbf{r} \in \mathcal{Q}(N)$. 
	Because a quantum channel is reversible if and only if the corresponding Kraus operators 
  can be written as a unitary matrix, this Pauli transfer matrix is indeed a representation of a unitary channel.
	
	Conversely, if $\hat{O}_\mathcal{E}$ is unitary, then it must be reversible with respect to every state in $\hat{\mathcal{Q}}(N)$, so $\forall i,\ |\lambda_i| > 0$. From Lem.~\ref{lem:transfer_matrix}, $\forall i,\ |\lambda_i| \leq 1$. It is then obvious that $\forall i, |\lambda_i| = 1$. 
	From the fourth statement of Lem.~\ref{lem:transfer_matrix} there is an eigenvalue decomposition $W = U \begin{pmatrix}
		e^{i\theta_1} & \cdots\\
		\vdots & e^{i\theta_2}\\
		 &  & \ddots
	\end{pmatrix}U^\dagger$. Then, $R^\dagger = R^{-1} = R^{T}$ because $R \in \mathbb{R}^{4^N-1}$. Therefore $\exists \mathcal{W} \in \mathrm{O}(4^N-1) \equiv \left\{Q \in \mathrm{GL}(4^N-1, \mathbb{R}) | Q^TQ = QQ^T = I\right\}$ such that $R_{\mathcal{E}} \in \mathcal{W}$.

\end{proof}

\begin{lem}{Coherence influx is orthogonal to the unit eigenvectors}

\label{lem:unit_spectral}
All eigenvectors $\vec\lambda$ corresponding to the eigenvalue $\lambda = 1$ are orthogonal to $\mathbf{b}$.
\end{lem}
\begin{proof}
First, we categorize the eigenvalues $\vec\lambda$ of $W$ and their special types of linear combinations as follows:
\begin{equation}
\begin{aligned}
	\Lambda_{=1}(W) &\equiv \left\{\vec\lambda\ \bigg|\  |\lambda| = 1\right\},\\
	\Lambda_{<1}(W) &\equiv \left\{\vec\lambda\ \bigg|\ |\lambda| < 1\right\}.\\
	\end{aligned}
\end{equation}

We omit the $(W)$ part of $\Lambda_*(W)$ for simplicity in the following discussions.

The general form of non-input-driven dynamics for $n \geq 1$ is 
	\begin{equation}
	\label{eqn:non-input-driven-dynamics}
			\mathbf{r}^{(n)} = \sum_{m = 0}^{n-1}W^m \mathbf{b} + W^n\mathbf{r}^{(0)}.
	\end{equation}
Given a Jordan canonical form of $W$ as
\begin{equation}
	W = MJM^{-1} = M\begin{pmatrix}
		\lambda_0&  \\
		 & \lambda_1 &  \\
		 & & \ddots \\
		 & & & \lambda_k & 1\\
		 & & & & \lambda_k\\
	\end{pmatrix}M^{-1},
\end{equation}
the general form of $J^n$ 
and $\sum_{m = n-1}^{0}J^\tau$ 
become
\begin{equation}
\begin{aligned}
	&J^n =
	\begin{pmatrix}
		\lambda_0^n&  \\
		 & \lambda_1^n &  \\
		 & & \ddots \\
		 & & & \lambda_k^n & n \lambda_k^{n-1}\\
		 & & & & \lambda_k^n\\
	\end{pmatrix},
	\end{aligned}
\end{equation}
and 
\begin{equation}
\label{eqn:infinity_power}
\begin{aligned}
	&\sum_{m=0}^nJ^m = \\
	&= \begin{pmatrix}
		\sum_{m} \lambda_0^m &  \\
		 & \sum_m \lambda_1^m &  \\
		 & & \ddots \\
		 & & & \lambda_k \frac{1 - \lambda_k^n}{1 - \lambda_k} & \lambda_k (1 - \lambda_k^{n-1}) - \frac{n \lambda_k^n}{1-\lambda_k}\\
		 & & & & \lambda_k \frac{1 - \lambda_k^n}{1 - \lambda_k}\\
	\end{pmatrix}.
	\end{aligned}
\end{equation}
Let $\ell$ be the maximum index such that $|\lambda_\ell| = 1$; then, it follows that when $n \to \infty$,
\begin{equation}
\label{eqn:infinity_jordan}
	\begin{aligned}
		J^n &\underset{n \to \infty}\to \begin{pmatrix}
		\lambda_0^n&  \\
		& \ddots & \\\
		& & \lambda_\ell^n &  \\
		 & & & 0 &  \\
		 & & & & \ddots \\
		\end{pmatrix},\\
		\sum_{m=0}^nJ^m &\underset{n \to \infty} \to\\
		&\begin{pmatrix}
		\sum_{m=0}^n \lambda_0^m &  \\
		& \ddots & \\
		& & \sum_{m=0}^n \lambda_\ell^m &  \\
		& & &  \ddots & \\
 		& & &   & \frac{\lambda_k}{1 - \lambda_k} & \lambda_k \\
 		& & &   &  & \frac{\lambda_k}{1 - \lambda_k} \\
			\end{pmatrix}.
	\end{aligned}
\end{equation}

The similarity transform $M$ and $M^{-1}$ are defined as 
\begin{equation}
\begin{aligned}
	M &= \begin{pmatrix}
		\vec\lambda_0 & \vec\lambda_1 & \hdots & \vec{\lambda_k}
	\end{pmatrix},\\
	M^{-1} &= \begin{pmatrix}
		\vec\lambda_0^{-1} &	 \vec\lambda_1^{-1} & \hdots & \vec\lambda_k^{-1}
	\end{pmatrix}^T,
\end{aligned}
\end{equation}
where some of the $\vec\lambda$s are generalized eigenvectors if $W$ is not diagonalizable. 
For all cases, the first term in Eq.~\eqref{eqn:non-input-driven-dynamics} becomes
\begin{equation}
\label{eqn:wb_infty}
\begin{aligned}
	\underset{n\to\infty}\lim \sum_{m = 0}^{n-1}W^m \mathbf{b} &= \lim_{n \to \infty}M \sum_{m=0}^{n-1}J^m M^{-1} \mathbf{b}\\
	&= \lim_{n\to \infty} \sum_{\lambda \in \Lambda_{=1}} \sum_{m = 0}^{n-1}\lambda^m \left(\vec\lambda^{-1} \cdot \mathbf{b}\right) \vec\lambda + \mathrm{const},
	\end{aligned}
\end{equation}
and the second term in Eq.~\eqref{eqn:non-input-driven-dynamics} becomes
\begin{equation}
\label{eqn:wr_infty}
\begin{aligned}
	\lim_{n \to \infty}\ W^n\mathbf{r}^{(0)} &=  \lim_{n \to \infty}M J^n M^{-1} \mathbf{r}^{(0)}\\
	 &= \lim_{n\to \infty} \sum_{\lambda \in \Lambda_{=1}} \lambda^n \left(\vec\lambda^{-1} \cdot \mathbf{r}^{(0)}\right) \vec\lambda
 \end{aligned}
\end{equation}
when $n \to \infty$. 

Suppose that there exists $\vec\lambda \in \Lambda_{=1}$ such that $\lambda = 1$; the right-hand side of Eq.~\eqref{eqn:wb_infty} diverges with respect to $n$ unless $\vec\lambda^{-1} \cdot \mathbf{b} = 0$ for all such $\vec\lambda$s. However, because the left-hand side of Eq.~\eqref{eqn:non-input-driven-dynamics} and the right-hand side of Eq.~\eqref{eqn:wr_infty} are bounded, $\vec\lambda^{-1} \cdot \mathbf{b} = 0$ for all such $\vec\lambda$s. It follows that $\vec\lambda \cdot \mathbf{b} = 0$ because $\vec\lambda = \vec\lambda^{-1}$ from the fourth statement of Lem.~\ref{lem:transfer_matrix}, which proves the lemma.
\end{proof}
\begin{cor}{Sufficient condition for positivity}

Suppose that we have a PTM $\hat{O} = \begin{pmatrix}
	1 & \mathbf{0}^T\\
	\mathbf{b} & W
\end{pmatrix} \in \hat{O}(N)$. If $\|\vec\lambda\cdot\mathbf{b}\| > 0$ for any eigenvector $\vec\lambda$ of $W$, then, no eigenvalue is equal to 1; hence, $(I - W)^T + (I - W)$ is positive definite.
\end{cor}
\begin{proof}
	From Lem.~\ref{lem:unit_spectral}, if $\|\vec\lambda\cdot\mathbf{b}\| > 0$ for an eigenvector $\vec\lambda$ of $W$, then, the corresponding eigenvalue is not equal to 1. Therefore, if $\|\vec\lambda\cdot\mathbf{b}\| > 0$ for all of the eigenvectors $\vec\lambda$ of $W$, then, $\mathrm{Ker}(I - W) = \emptyset$, and the real parts of all of the eigenvalues of $I - W$ are positive because of the second statement of Lem.~\ref{lem:transfer_matrix}. This implies that all of the eigenvalues of $(I - W)^T + (I -W) = (I - W)^\dagger + (I - W)$ is positive real, which proves the corollary.
\end{proof}

\begin{cor}{Non-vanishing coherence influx ensures the existence of positivity-proved invariant subspace}
\label{cor:non-vanishing-coherenc-influx}

Suppose that we have a PTM $\hat{O} = \begin{pmatrix}
	1 & \mathbf{0}^T\\
	\mathbf{b} & W
\end{pmatrix} \in \hat{O}(N)$. If $\|W\mathbf{b}\| > 0$, then, there exists an invariant subspace $\mathcal{Q}_\mathbf{b}$ of $W$ such that $\mathbf{b} \in \mathcal{Q}_\mathbf{b}$ and $\vec\lambda \notin \mathcal{Q}_\mathbf{b}$ for all eigenvectors $\vec\lambda$ of $W$ that have a corresponding eigenvalue of 1. That is, $W$ does not have an eigenvalue one in $\mathcal{Q}_\mathbf{b}$; hence, $(I - W)^T + (I - W)$ is positive definite in $\mathcal{Q}_\mathbf{b}$. Namely, $\mathbb{P}_{\mathcal{Q}_\mathbf{b}}\left[(I - W)^T + (I - W)\right]$ is positive definite where $\mathbb{P}_{\mathcal{Q}_\mathbf{b}}$ is a projector from $\mathcal{Q}(N)$ onto $\mathcal{Q}_\mathbf{b}$.
	
\end{cor}
\begin{proof}
For any eigenvectors $\vec\lambda$ of $W$, $\vec\lambda \cdot \mathbf{b} = 0$ if its corresponding eigenvalue $\lambda = 1$. Therefore, there exists at least one eigenvector $\vec\lambda_\perp$ of $W$ such that the corresponding eigenvalue $\lambda_\perp\neq 1$. Let $\mathcal{Q}_{\mathbf{b}_\perp} \equiv \mathrm{span}(\{\vec\lambda | \lambda = 1\})$, $\mathcal{Q}_{\mathbf{b}} \equiv \mathcal{Q}_{\mathbf{b}_\perp}^\perp$, and $\mathcal{G}_{\mathcal{Q}_{\mathbf{b}}} \equiv \mathbb{P}_{\mathcal{Q}_{\mathbf{b}}}(I - W)$, where $\mathbb{P}_{\mathcal{Q}_{\mathbf{b}}}$ is a projector onto $\mathcal{Q}_{\mathbf{b}}$. Then $\mathbb{P}_{\mathcal{Q}_\mathbf{b}}\left[(I - W)^T + (I - W)\right]$ is positive definite because no eigenvalues of $\mathbb{P}_{\mathcal{Q}_{\mathbf{b}}}(I - W)$ are equal to 1.
\end{proof}

\section{Preparation for the proof of theorems}
\begin{lem}{Equivalent condition for strict decrease of distance between coherence vectors}
\label{lem:equiv_strict_dec_hs}

	Let $\mathbf{r}^{(t)}$ and $\mathbf{r}^{'(t)}$ be two different system states at time $t$; then, $\|\delta_{t} \| \equiv \|\mathbf{r}^{(t)} - \mathbf{r}^{'(t)}\|$ strictly decreases with respect to $t$ if and only if $\mathcal{G}(\mathbf{u}_t)^T + \mathcal{G}(\mathbf{u}_t)$ is positive definite for every $\mathbf{u}_t \in \mathcal{X}$.
\end{lem}
\begin{proof}
	$\delta_{t}$ holds the following equation:
\begin{equation}
\begin{aligned}
	\delta_{t+1} &\equiv \mathbf{r}^{(t+1)} - \mathbf{r}^{'(t+1)}\\
	&= WR(\mathbf{u}_t)\left(\mathbf{r}^{(t)} - \mathbf{r}^{'(t)}\right)\\
	&= WR(\mathbf{u}_t)\delta_t.
\end{aligned}
\end{equation}

Therefore, if $D$ is a differentiation operator, then, 
\begin{equation}
\begin{aligned}
	D\delta_t &\equiv \delta_{t+1} - \delta_t\\
	 &= (WR(\mathbf{u}_t) - I ) \delta_t\\
	 &= -\mathcal{G}(\mathbf{u}_t)\delta_t.
	\end{aligned}
\end{equation}
Then, taking the differentiation of $\|\delta_t\|^2 = \delta_t^T\delta_t$ is read the following equation:
\begin{equation}
\label{eqn:delta_t_difference}
\begin{aligned}
	D\|\delta_t\|^2 &= D\delta_t^T \delta_t + \delta_t^T D\delta_t \\
	&= -\delta_t^T \left(\mathcal{G}(\mathbf{u}_t)^T + \mathcal{G}(\mathbf{u}_t)\right) \delta_t.
\end{aligned}
\end{equation}
Because $\mathcal{G}(\mathbf{u}_t)^T + \mathcal{G}(\mathbf{u}_t)$ is symmetric, and thus has only real eigenvalues, Eq~\eqref{eqn:delta_t_difference} implies that $\|\delta_t\|$ strictly decreases for all $\mathbf{r}^{(t)}$ and $\mathbf{r}^{'(t)}$ if and only if $\mathcal{G}(\mathbf{u}_t)^T + \mathcal{G}(\mathbf{u}_t)$ is positive definite for every $\mathbf{u}_t \in \mathcal{X}$, which proves the lemma.
\end{proof}

\begin{lem}{Sufficient conditions for a strict decrease of distance between coherence vectors}
\label{lem:suffice_strict_dec_hs}

	Let $\mathbf{r}^{(t)}$ and $\mathbf{r}^{'(t)}$ be two different system states at time $t$; then, $\|\delta_{t} \| \equiv \|\mathbf{r}^{(t)} - \mathbf{r}^{'(t)}\|$ strictly decreases with respect to $t$ if $WR(\mathbf{u_t})$ does not have an eigenvalue 1 and $\mathcal{G}(\mathbf{u_t})$ is diagonalizable for every $\mathbf{u}_t \in \mathcal{X}$.
\end{lem}
\begin{proof}
	Because  $WR(\mathbf{u_t})$ does not have an eigenvalue 1 for every $\mathbf{u}_t\in \mathcal{X}$, all of the eigenvalues of $\mathcal{G}(\mathbf{u}_t) = \left(I - WR(\mathbf{u}_t)\right)$ have positive real parts. If $\mathcal{G}(\mathbf{u}_t)$ is diagonalizable as $UD^*U^\dagger$, then, $\mathcal{G}(\mathbf{u}_t)^T = \mathcal{G}(\mathbf{u}_t)^\dagger  =(UD^*U^\dagger)^\dagger = UDU^\dagger$ because $\mathcal{G}(\mathbf{u}_t)$ is real. Therefore, $\mathcal{G}(\mathbf{u}_t) + \mathcal{G}(\mathbf{u}_t)^T = 2U\mathrm{Re}(D)U^\dagger$, where $\mathrm{Re}(D)$ is the real part of $D$. This implies that $\mathcal{G}(\mathbf{u}_t) + \mathcal{G}(\mathbf{u}_t)^T$ is positive definite. Therefore,  $\|\delta_{t} \|$ strictly decreases with respect to $t$ from Lem.~\ref{lem:equiv_strict_dec_hs}, which proves the lemma.
\end{proof}

\section{Proofs of theorems}
\subsection{Proof of Rem.~\ref{rem:unital}}
\begin{proof}
	It is obvious that $|I\rangle\rangle = \begin{pmatrix}
		1\\
		\mathbf{0}
	\end{pmatrix}$. Therefore, a state update by a PTM $\hat{O} = \begin{pmatrix}
		1 & \mathbf{0}^T\\
		\mathbf{b} & W
	\end{pmatrix}$ reads $\hat{O}|I\rangle\rangle = \mathbf{b}$, so $\mathbf{b}=\mathbf{0}$ if $\hat{O}$ is unital. Conversely, if $\mathbf{b}=0$, then, $\hat{O}|I\rangle\rangle = |I\rangle\rangle$, which imples that $\hat{O}$ is unital.
\end{proof}

\subsection{Proof of Lem.~\ref{lem:mq_bloch}}
\begin{proof}
	
	1. A set of all normalized $N$-qubit Pauli strings $\left\{\frac{P_i}{\sqrt{2^N}}\right\}$ forms an orthonormal basis set of $\mathbb{C}^{2^N}$ under the Hilbert–Schmidt inner product $\langle A, B\rangle = \mathrm{tr}(A^\dagger B)$. Therefore, $\rho = \frac{1}{2^N}\sum_i P_i \mathrm{tr}(\rho P_i)$ is an orthonormal basis expansion of $\rho$.
	
	2. Given a unitary transformation $U$ in a density matrix formulation, 
	\begin{equation}
	\label{eqn:unitary_transformed_pauli_basis}
		\|\mathbf{r}\|^2 = \sum_i \mathrm{tr}(P_i U \rho U^\dagger)^2 = \sum_i \mathrm{tr}(U^\dagger P_i U \rho)^2.
	\end{equation}
	Because $\left\{\frac{U^\dagger P_i U}{\sqrt{2^N}}\right\}$ also forms an orthonormal basis set of $\mathbb{C}^{2^N}$ under the Hilbert–Schmidt inner product, the representation in Eq.~\eqref{eqn:unitary_transformed_pauli_basis} is a change of orthonormal basis set from the original Pauli strings without scaling. Therefore, $\|c_N\mathbf{r}\|$ is unchanged.
	
	3. First, as a reminder, any mixed state density matrix $\rho_{mix}$ can be written as a weighted sum of pure state density matrices $\rho_{pure}^{(i)}$. That is,
		\begin{equation}
			\forall \rho_{mix},\ \exists \{\rho_{pure}^{(i)}\}\text{ s.t. } \rho_{mix} = \sum_i p_i \rho_{pure}^{(i)}\text{ and } \sum_i p_i = 1.
		\end{equation}
		By the Def.~\ref{dfn:multi-qubit-bloch}, the multi-qubit Bloch vector $c_N\mathbf{r}_{mix}$ can also be written as a sum of pure state Bloch vectors $\{c_N \mathbf{r}_{pure}^{(i)}\}$. That is, 
		\begin{equation}
		\label{eqn:bloch_linear_interpolation}
			\forall \mathbf{r}_{mix},\ \exists \{\mathbf{r}_{pure}^{(i)}\}\text{ s.t. } \mathbf{r}_{mix} = \sum_i p_i \mathbf{r}_{pure}^{(i)}\text{ and } \sum_i p_i = 1.
		\end{equation}
		Using the triangle inequality, 
		\begin{equation}
			\forall \mathbf{r}_{mix},\ \|\mathbf{r}_{mix}\| \leq \sum_i p_i \|\mathbf{r}_{pure}^{(i)}\| = \|\mathbf{r}_{pure}\|.
		\end{equation}
		Because any quantum system has only mixed or pure states, $\|\mathbf{r}\|$ is maximized if and only if it is pure.
	Then, from the definition in Eq.~\eqref{eqn:q_n}, 
	\begin{equation}
		\|\mathbf{r}\|^2 = \sum_{i\geq 1} \mathrm{tr}\left(P_{i}\rho\right)^2.
	\end{equation} 
	One of the $N$-qubit pure states is
	\begin{equation}
		\rho = |P_{k^{\otimes N}}\rangle \langle P_{k^{\otimes N}}|\text{ s.t }P_{k^{\otimes N}}|P_{k^{\otimes N}}\rangle = |P_{k^{\otimes N}}\rangle,\quad k \geq 1,
	\end{equation} where $P_{k^{\otimes N}} \equiv \bigotimes_i \sigma_k^{(i)}$. 
  
 For example, that is $|0\rangle^{\otimes N}$ in state vector representation. In this case, $\mathrm{tr}\left(P_{i}\rho\right) = 0$ if any of $P_i$ contains $\sigma_l$ s.t. $l \notin \{0,\ k\}$. The number of non-zero terms is then $\left|\{\sigma_0, \sigma_k\}^{\otimes N}\right| - \left|\{\sigma_0^{\otimes N}\}\right| = 2^N - 1$. For every such term $\mathrm{tr}\left(P_i\rho\right)^2 = 1$, so $\|\mathbf{r}\|^2 \leq 2^N-1$, which proves the proposition.
 \\
	4. Because a linear combination of two density matrices, $\rho_1$ and $\rho_2$, that can be written as $\rho_3 = p\rho_1 + (1-p)\rho_2$ where $p \in [0, 1]$, is also a density matrix, a linear combination of coherence vector is also a proper coherence vector. Specifically, a linear combination of a coherence vector $\mathbf{r}$ with a completely mixed state that can be written as follows:
	\begin{equation}
		\mathbf{r}' = p\mathbf{r} + (1-p) \mathbf{0}= p\mathbf{r} \in \mathcal{Q}(N).
	\end{equation}
	Therefore, $0\leq \forall p \leq 1$, $p\mathbf{r} \in \mathcal{Q}(N)$.

\end{proof}

\subsection{Proof of Lem.~\ref{lem:suffice_trad_esp}}

\begin{proof}
	For QRC, $\mathcal{S} = \mathcal{Q}(N)$ and $f(\{\mathbf{u}_\tau\}_{\tau \leq t} ; s_0) = \mathbf{r}^{(t)}$ are the specific forms of variables in Eq.~\eqref{eqn:esp_ext2}. By the state update rule of QRC in Eq.~\eqref{eqn:qrc_gen_form}, the state difference can be defined as follows:
	\begin{equation}
	\label{eqn:state_diff_qrc}
	\begin{aligned}
		 \left\|\Delta \mathbf{r}^{(t)}\left(\mathbf{r}^{(0)},\mathbf{r}^{(0)'}\right)\right\|  &\equiv \|\mathbf{r}^{(t)} - \mathbf{r}^{'(t)}\| \\
				& = \left\|\left(\prod_{\tau \leq t}WR(\mathbf{u}_{\tau})\right)\left({\mathbf{r}}^{(0)} - {\mathbf{r}}^{(0)'} \right) \right\|\\
			\end{aligned}
	\end{equation}
	If Eq.~\ref{eqn:spectral_converge} holds, then, the last form of Eq.~\eqref{eqn:state_diff_qrc} can be upper bounded by
	\begin{equation}
	\begin{aligned}
			\sigma_{\mathrm{max}}\left(\prod_{\tau \leq t}WR(\mathbf{u}_{\tau})\right) \|\mathbf{r}^{(0)} - \mathbf{r}^{(0)'}\| \underset{t \to \infty} \to 0,
	\end{aligned}
	\end{equation}
	for any $\{\mathbf{u}_t \in \mathcal{X}\}$. Therefore, for any $\{\mathbf{u}_t \in \mathcal{X}\}$ 
	\begin{equation}
				\forall (\mathbf{r}^{(0)}, \mathbf{r}^{'(0)}),\ \|\mathbf{r}^{(t)} - \mathbf{r}^{'(t)}\| \underset{t \to \infty}\to  0,
	\end{equation}
	which proves the lemma.
\end{proof}

\subsection{Proof of Prop.~\ref{prop:coherence-influx}}
\begin{proof}
If $\mathbf{b} = 0$, then, the state update rule of QRC becomes
	\begin{equation}
			\mathbf{r}^{(t)} = \left(\prod_{\tau \leq t}WR(\mathbf{u}_{\tau})\right)\mathbf{r}^{(0)},
	\end{equation}
	and the denominator of the quantity below becomes a constant 0 when $\mathbf{r}^{(0)} = \mathbf{0}$.
\begin{equation}
\label{eqn:b_necessary}
	\frac{ \|\mathbf{r}^{(t)} - \mathbf{r}^{'(t)}\|}{\sqrt{\min\left(\overline{\mathrm{Var}}_w^t[\mathbf{r}^{(t)}; \mathbf{r}^{(0)}], \overline{\mathrm{Var}}_w^t[\mathbf{r}^{(t)}; \mathbf{r}^{'(0)}]\right)}}
\end{equation}
Therefore, Eq.~\eqref{eqn:b_necessary} diverges for some $\mathbf{r}^{'(0)}$, which shows that the QRC does not satisfy the nonstationary ESP.
\end{proof}
\subsection{Proof of Thm.~\ref{thm:sufficient_ns_esp}}
\begin{proof}
		By Lem.~\ref{lem:suffice_trad_esp}, condition 3 yields the below property.
	\begin{equation}
	\label{eqn:contracting-wrr}
	\begin{aligned}
		&\forall \{\mathbf{u}_t\},\ \forall (\mathbf{r}^{(0)}, \mathbf{r}^{'(0)}),\\
		 &\left\|\Delta \mathbf{r}^{(t)}\left(\mathbf{r}^{(0)},\mathbf{r}^{(0)'}\right)\right\| \equiv \|\mathbf{r}^{(t)} - \mathbf{r}^{'(t)}\|  \underset{t \to \infty} \to 0\\
	\end{aligned}
	\end{equation}
	In addition, let
	\begin{equation}
	\label{eqn:w_tilde_b_t}
		\begin{aligned}
			\mathbf{b}^{(t)} \equiv \sum^{t-1}_{\tau=0}\left(\prod_{1\leq n \leq \tau}WR(\mathbf{u}_{t - n})\right)\mathbf{b}.
		\end{aligned}
	\end{equation}

	Then, the general form of input-driven state evolution is
	\begin{equation}
	\begin{aligned}
		\mathbf{r}^{(t)} &= \mathbf{b}^{(t)} + \prod_{\tau \leq t}WR(\mathbf{u}_{\tau})\mathbf{r}^{(0)}\\
		&\underset{t \to \infty} \to \lim_{t\to \infty}  \mathbf{b}^{(t)}.
		\end{aligned}
	\end{equation}
	Here, the second term in the right-hand side of the first row vanishes when $t \to \infty$ because 
	\begin{equation}
	\begin{aligned}
		\left\|\prod_{\tau \leq t}WR(\mathbf{u}_{\tau})\mathbf{r}^{(0)}\right\| &= \left\|\Delta \mathbf{r}^{(t)}\left(\mathbf{r}^{(0)},\mathbf{0}\right)\right\|\\
		&\underset{t \to \infty} \to 0.
		\end{aligned}
	\end{equation}
	as implied in Eq.~\eqref{eqn:contracting-wrr}.
	
	$\mathbf{b}^{(t)}$ is invariant under input-driven dynamics of the next time step $t+1$ if and only if
	\begin{equation}
	\label{eqn:b_to_fixed_point}
	\begin{aligned}
		\mathbf{b}^{(t + 1)} - \mathbf{b}^{(t)} &= \mathbf{b} + WR(\mathbf{u}_t)\mathbf{b}^{(t)} - \mathbf{b}^{(t)}\\
		&= \mathbf{b} - \boldsymbol{\left(}I - WR(\mathbf{u}_t)\boldsymbol{\right)}\mathbf{b}^{(t)}\\
		&= 0.
	\end{aligned}
	\end{equation}
Because of condition 1, the inverse matrix $\left(I - WR(\mathbf{u}_t)\right)^{-1}$ always exists. Therefore, Eq.~\eqref{eqn:b_to_fixed_point} is equivalent to
\begin{equation}
\label{eqn:i_wr_inverse_b}
	\boldsymbol{\left(}I - WR(\mathbf{u}_t)\boldsymbol{\right)}^{-1}\mathbf{b} = \mathbf{b}^{(t)}.
\end{equation}
Let us denote $\mathbf{b}^{(t)} = \mathbf{b}^{(t)}(\mathbf{u}_t)$ to indicate that $\mathbf{b}^{(t)}$ depends on $\mathbf{u}_t$. 
Because of the condition 2, $\mathbf{b}^{(t)}(\mathbf{u}_t)$ is an injective from $\mathcal{X}$ to $\mathcal{Q}(N)$. That is, for any $\mathbf{u}, \mathbf{v} \in \mathcal{X}$ such that $\|\mathbf{u}_t - \mathbf{v}_t\| > \delta > 0$, there exists $\epsilon > 0$ such that $\|\mathbf{b}^{(t)}(\mathbf{u}_t) - \mathbf{b}^{(t)}(\mathbf{v}_t)\| > \epsilon$. Therefore, even if there exists one input $\mathbf{u}_t$ such that Eq.~\eqref{eqn:i_wr_inverse_b} holds, it does not hold for the other input $\mathbf{v}_t \neq \mathbf{u}_t$.
	Because $\underset{t\to \infty}{\liminf}\ \mathrm{Var}_w^t[\{\mathbf{u}_\tau\}] > 0$ indicates that at least one pair of different inputs in time window of size $w$, there must be at least one pair of different $\mathbf{b}^{(t)}(\mathbf{u}_t)$ in time window of size $w$. This implies that $\mathrm{Var}_w^t(\mathbf{r}^{(t)}; \mathbf{r}^{(0)}) > 0$ for any $\mathbf{r}^{(0)}$. Therefore, we can conclude that the nonstationary ESP holds under the unitary input encoding $R$. That is, 
	\begin{equation}
	\label{eqn:qrc_ns_esp}
	\begin{aligned}
		\lim_{t \to \infty}\frac{ \|\mathbf{r}^{(t)} - \mathbf{r}^{'(t)}\|}{\sqrt{\min\left(\overline{\mathrm{Var}}_w^t[\mathbf{r}^{(t)}; \mathbf{r}^{(0)}], \overline{\mathrm{Var}}_w^t[\mathbf{r}^{(t)}; \mathbf{r}^{'(0)}]\right)}}  = 0.
		\end{aligned}
	\end{equation}

\end{proof}

\subsection{Proof of Lem.~\ref{lem:monotonic_hs_dist}}
\begin{proof}
This is a composition of Lem.~\ref{lem:equiv_strict_dec_hs} and Lem.~\ref{lem:suffice_strict_dec_hs}.
\end{proof}
\subsection{Proof of Prop.~\ref{prop:suffice_convergence}}
\begin{proof}
Conditions 1 and 2 are direct consequences of Lem.~\ref{lem:monotonic_hs_dist}.

For condition 3, we follow the proof of Thm.~4.1 in \cite{yildiz2012re}. Here, condition 3 ensures that for all $\mathbf{u}_t \in \mathcal{X}$, there exists a single positive definite symmetric matrix $P \succ 0$ such that $\left(WR(\mathbf{u}_t)\right)^TPWR(\mathbf{u}_t) - P \prec 0$. Suppose we have two system states $\mathbf{r}^{(t)}$ and $\mathbf{r}^{'(t)}$ at time $t$, that were initialized at time $t=0$ with different initial states, and evolved with identical input sequences up to time $t$. We have
\begin{equation}
\begin{aligned}
	\Delta \mathbf{r}^{(t+1)} &\equiv \mathbf{r}^{(t+1)} - \mathbf{r}^{'(t+1)}\\
	& = WR(\mathbf{u}_t)\left(\mathbf{r}^{(t)} - \mathbf{r}^{'(t)}\right)\\
	& = WR(\mathbf{u}_t)\Delta \mathbf{r}^{(t)}.
\end{aligned}
\end{equation}
We need to show that $\left\|\Delta \mathbf{r}^{(t)}\right\| \underset{t \to \infty}\to 0$ for any $\{\mathbf{u}_t\}$ under condition 3. For $x \in \mathbb{R}^{4^N-1}$, let $F(x) = x^T P^* x$ for some symmetric $P^* \succ 0$; then, $F(x) = 0$ if and only if $x = 0$. Because $F(x) \geq 0$, it is sufficient to show that $F(\Delta\mathbf{r}^{(t)})$ is a strongly decreasing sequence of $t$ to prove the convergence of $\|\Delta\mathbf{r}^{(t)}\|$ to 0 as $t \to \infty$. From here we write $\Delta\mathbf{r}^{(t)}$ as $z_t$ for simplicity. Let us define 
\begin{equation}
	\Delta F_{t+1} \equiv F(z_{t+1}) - F(z_t).
\end{equation}
A simple calculation reads
\begin{equation}
	\Delta F_{t+1} = z_t^T\left[\boldsymbol{\left(}WR(\mathbf{u}_t)\boldsymbol{\right)}^TP^*WR(\mathbf{u}_t) - P^*\right]z_t.
\end{equation}
We know that for any $\{\mathbf{u}_t\}$, there exists respective $P \succ 0$ such that $\boldsymbol{\left(}WR(\mathbf{u}_t)\boldsymbol{\right)}^TPWR(\mathbf{u}_t) - P \prec 0$. Therefore, by taking $P^* = P$, the strong decrease of $F(z_t)$ for each $t$ is proved. This results in $F(z_t) \underset{t \to \infty} \to 0$ and that implies $\|\Delta\mathbf{r}^{(t)}\| \underset{t \to \infty}\to 0$ for any $\{\mathbf{u}_t\}$.

If $\lim_{t\to \infty}\sigma_{\mathrm{max}}\boldsymbol{(}\prod_t WR(\mathbf{u}_t)\boldsymbol{)} > 0$, then, there exists initial difference $\Delta\mathbf{r}^{(0)} \neq 0$ such that $\|\Delta\mathbf{r}^{(t)}\| > 0$ when $t \to \infty$. However, it contradicts the conclusion of the convergence of $\|\Delta\mathbf{r}^{(t)}\|$ above. Therefore, $\lim_{t\to \infty}\sigma_{\mathrm{max}}\boldsymbol{(}\prod_t WR(\mathbf{u}_t)\boldsymbol{)} = 0$ under condition 3 is proved.

If condition 4 holds---that is, there exists a matrix norm $\|\cdot\|_M$ such that $\|WR(\mathbf{u}_t)\|_M < 1$ for every $\mathbf{u}_t \in \mathcal{X}$---then, by the sub-multiplicativity of matrix norms reads $\left\|\prod_{\tau=0}^t \boldsymbol{\left(}WR(\mathbf{u}_\tau)\boldsymbol{\right)}\right\|_M \leq \prod_{\tau=0}^t \left\|WR(\mathbf{u}_\tau)\right\|_M \underset{t\to\ \infty}{\to} 0 $ for every $\{\mathbf{u}_t \in \mathcal{X}\}_t$. Because for any different matrix norms $\|\cdot\|_A$ and $\|\cdot\|_B$, there exists $\alpha, \beta > 0$ such that $\alpha\|W\|_A\leq \|W\|_B \leq \beta \|W\|_A$, $\|W\|_B = 0$ implies $\|W\|_A = 0$. Using this fact, if there exists a matrix norm $\|\cdot\|_M$ that satisfies the above convergence, the spectral norm also converges in the same situation. Therefore, $\sigma_{\mathrm{max}}\boldsymbol{(}\prod_t WR(\mathbf{u}_t)\boldsymbol{)} \underset{t \to \infty}\to 0$ is implied, which proves the proposition.
\end{proof}

\subsection{Proof of Rem.~\ref{rem:spectral_norm_reset}}
\begin{proof}
Because there is no dissipation except for the reset operations, the system dynamics is unitary. Let us remember that multiplication of any orthogonal matrix does not change the spectral norm. Therefore, the spectral norm of the total system dynamics is the spectral norm of the reset operation irrespective of input encoding. A single-qubit reset operation has a PTM of form $\Gamma(1)_Z \equiv \begin{pmatrix}
		1 & \mathbf{0}^T\\
		\mathbf{1}_Z & \mathrm{O}
	\end{pmatrix}$, where $\mathrm{O} \in \mathbb{R}^{3 \time 3}$ denotes an all-zero matrix, and $\mathbf{1}_Z \equiv \begin{pmatrix}
		0 & 0 & 1
	\end{pmatrix}^T$. Overall encoding can be written as $\hat{E} \equiv \left(\Gamma(1)_Z^{\otimes M}I^{\otimes N-M}\right)$, where $M$ denotes the number of qubits to be replaced within the $N$-qubits system, and the indices of qubits are sorted so that first $M$-qubits are reset. Let $\hat{E} = \begin{pmatrix}
		1 & 0\\
		\mathbf{b}_E & E
	\end{pmatrix}$.
	
	We will prove $\sigma_\mathrm{max}(E) = 2^\frac{M}{2}$ with an induction with respect to $M$ when $M < N$. If $M=1$ and $N\geq 2$, the matrix form of $\hat{E}$, which we write as $\hat{E}_{1, N}$, can be written as follows:
	\begin{equation}
	\begin{aligned}
\hat{E}_{1, N} = \begin{pmatrix}
I_{4^{N-1}} & \mathbf{0}_{3\cdot 4^{N-1}}^T\\
\mathbf{0}_{2\cdot 4^{N-1}} & \mathrm{O}\\
I_{4^{N-1}} & \mathrm{O}\\	
\end{pmatrix}
\end{aligned}, 
\end{equation}
where, $I_k$ denotes an identity matrix of $k$ dimensions, $\mathbf{0}_k$ denotes an all-zero column vector of $k$ dimensions, and $\mathrm{O}$ denotes an all-zero matrix of appropriate dimensions. Therefore, the length of any of the one-hot vector $\mathbf{r}_{\mathrm{oh}}$ that have value 1 in one of the first $4^{N-1}-1$ indices will be expanded to two, which implies that $\sigma_{\mathrm{max}}(E) = \sqrt{2} = 2^{\frac{M}{2}}$.

Suppose that $\sigma_\mathrm{max}(E) = 2^\frac{K}{2}$ holds for $K$-qubits reset. Then, $\hat{E}_{M+1, N+1}$ can be written as follows:
\begin{equation}
\begin{aligned}
	E_{K+1, N+1} &= \Gamma(1)_Z \otimes \hat{E}_{K, N}\\
	&= \begin{pmatrix}
		\hat{E}_{K, N} & \mathbf{0}_{3\cdot 4^{N}}^T\\
		\mathbf{0}_{2\cdot 4^{N}} & \mathrm{O}\\
		\hat{E}_{K, N} & \mathrm{O}
	\end{pmatrix},
	\end{aligned}
\end{equation}
which implies that $\sigma_\mathrm{max}^2(E_{M, N}) = 2 \sigma_\mathrm{max}^2(E_{K, N})$, so $\sigma_\mathrm{max}(E_{M, N}) = \sqrt{2} \sigma_\mathrm{max}^2(E_{K, N}) = 2^{\frac{M}{2}}$.
Because we already proved the fact for $\hat{E}_{1, N}$ with any $N \geq 1$, this also implies that $\sigma_\mathrm{max}(E) = 2^\frac{M}{2}$ for any $M$-qubits reset when $M < N$ within $N$-qubits system, which proves the remark.
\end{proof}

\subsection{Proof of Lem.~\ref{lem:positiveity_subspace}}
\begin{proof}
	Suppose that we have a PTM $\hat{O} = \begin{pmatrix}
	1 & \mathbf{0}^T\\
	\mathbf{b} & W
\end{pmatrix} \in \hat{O}(N)$. If $\|W\mathbf{b}\| > 0$, then, there exists an invariant subspace $\mathcal{Q}_\mathbf{b}$ of $W$ such that $\mathbf{b} \in \mathcal{Q}_\mathbf{b}$ and $\vec\lambda \notin \mathcal{Q}_\mathbf{b}$ for all eigenvectors $\vec\lambda$ of $W$ that have a corresponding eigenvalue of 1. That is, $W$ does not have an eigenvalue one in $\mathcal{Q}_\mathbf{b}$; hence, $(I - W)^T + (I - W)$ is positive definite in $\mathcal{Q}_\mathbf{b}$. Namely, $\mathbb{P}_{\mathcal{Q}_\mathbf{b}}\left[(I - W)^T + (I - W)\right]$ is positive definite where $\mathbb{P}_{\mathcal{Q}_\mathbf{b}}$ is a projector from $\mathcal{Q}(N)$ onto $\mathcal{Q}_\mathbf{b}$.
\end{proof}

\subsection{Proof of Cor.~\ref{cor:subspace_ns_esp_suffice}}
\begin{proof}
	If $\|W\mathbf{b}\| > 0$, there exist invariant subspaces of $W$: $\mathcal{Q}_\mathbf{b}$ and $\mathcal{Q}_\mathbf{b}^\perp$ such that $\mathcal{Q}_\mathbf{b} \oplus \mathcal{Q}_\mathbf{b}^\perp = \mathcal{Q}(N)$. If there exists an input encoding $\hat{R}$ such that $R$ also has $\mathcal{Q}_\mathbf{b}$ and $\mathcal{Q}_\mathbf{b}^\perp$ as its invariant subspaces, then the state update rule can be written as the direct sum below.
	\begin{equation}
	\begin{aligned}
		\mathbf{r}_\mathbf{b}^{\perp (t+1)} \oplus \mathbf{r}_\mathbf{b}^{(t+1)} = 
			\mathbb{P}_{\mathcal{Q}_b}WR(\mathbf{u}_t)\mathbf{r}^{(t)} \oplus \mathbb{P}_{\mathcal{Q}_b^\perp}WR(\mathbf{u}_t)\mathbf{r}^{(t)} + \mathbf{b}.
		\end{aligned}
	\end{equation}
Because $\mathbf{b} \in \mathcal{Q}_\mathbf{b}$, a state update rule in the subspace $\mathcal{Q}_\mathbf{b}$ become as follows:
\begin{equation}
\label{eqn:subspace_update_rule}
	[\mathbb{P}_{\mathcal{Q}_b}\mathbf{r}_\mathbf{b}^{(t+1)}]_{\mathcal{Q}_\mathbf{b}} = [\mathbb{P}_{\mathcal{Q}_b}WR(\mathbf{u}_t)]_{\mathcal{Q}_\mathbf{b}}[\mathbb{P}_{\mathcal{Q}_b}\mathbf{r}^{(t)}]_{\mathcal{Q}_\mathbf{b}} + [\mathbf{b}]_{\mathcal{Q}_\mathbf{b}},
\end{equation}
where we denote a representation of any matrix or vector in $\mathcal{Q}_\mathbf{b}$ by its orthonormal basis set as $[\cdot]_{\mathcal{Q}_\mathbf{b}}$.

Therefore, the same discussion as the proof of Thm.~\ref{thm:sufficient_ns_esp} reads the sufficient condition of this QRC's subspace nonstationary ESP as all of the following statements:
\begin{enumerate}
 	\item Inverse matrices $[\mathcal{G}_\mathbf{b}(\mathbf{u}_t)]_{\mathcal{Q}_\mathbf{b}}^{-1}$ always exist for all $\mathbf{u}_t \in \mathcal{X}$, where  $\mathcal{G}_\mathbf{b}(\mathbf{u}_t) \equiv \mathbb{P}_{\mathcal{Q}_\mathbf{b}}\left(I - WR(\mathbf{u}_t)\right)$ and $[\cdot]_{{\mathcal{Q}_\mathbf{b}}}$ denotes a representation of a vector or a matrix in $\mathcal{Q}_\mathbf{b}$ with its orthonormal basis set.
 	\item $\mathbf{u}_t \mapsto \mathcal{G}_\mathbf{b}^{-1}(\mathbf{u}_t)\mathbf{b}$ is an injective map from $\mathcal{X}$ to $\mathcal{Q}_\mathbf{b}$.
 	\item $s_t^{\mathbf{b}}\left(W, R; \{\mathbf{u}_t\}\right) \equiv \sigma_{\mathrm{max}}\left(\mathbb{P}_{\mathcal{Q}_\mathbf{b}}\prod_t WR(\mathbf{u}_t)\right) \underset{t\to \infty}\to 0$ for any $\{\mathbf{u}_t \in \mathcal{X}\}$.
\end{enumerate}

\end{proof}

\subsection{Proof of Thm.~\ref{thm:mrc-ns-esp}}
\begin{proof}
Let us again note that $\mathcal{X} = [-1, 1]$. Let two different state sequences $\{x_t\}$ and $\{x_t'\}$ be initiated with different initial states $x_0$ and $x_0'$, respectively, and define
\begin{equation}
\begin{aligned}
	|\Delta x_t| \equiv |x_t - x_t'|.
\end{aligned}
\end{equation}
From Eq.~\eqref{eqn:mrc-general}, 
\begin{equation}
\label{eqn:mrc_converge}
	|\Delta x_t| = \left|a^{t+1}\prod_{\tau=0}^{t}u_\tau\right||\Delta x_0| \sim \mathrm{O}(|a|^t).
\end{equation}
In addition,
\begin{equation}
	x_{t+1} - x_t  = (a u_t - 1)x_t + b.
\end{equation}

First, if $a=0$, $x_t = b$ for all $t$. Therefore, nonstationary ESP does not hold.

Second, if $|a| > 1$, $|\Delta x_t| \underset{t \to \infty}\to \infty$ for any $\{u_t\}$ provided that $|\Delta x_0| > 0$. Therefore, nonstationary ESP does not hold.

Third, if $0 < |a| < 1$, then $x_{t+1} - x_t  = 0$ implies $x_t = \frac{b}{1 - au_t}$. In this case, $|b| > 0$ and $\mathrm{Var}_w(\{u_t\}) > 0$ induce $\mathrm{Var}_w(\{x_t\}) > 0$. Because Eq.~\eqref{eqn:mrc_converge} ensures state-difference decay, nonstationary ESP holds.

Finally, if $|a| = 1$, then $x_{t+1} - x_t  = 0$ only if $u_t = \mathrm{sign}(a)$. Therefore, $\mathrm{Var}_w(\{u_t\}) > 0$ ensures $\mathrm{Var}_w(\{x_t\}) > 0$, provided that $|b| > 0$. However, because Eq.~\eqref{eqn:mrc_converge} does not converge, nonstationary ESP does not hold.

Overall, if $|b| > 0$ and  $0 < |a| < 1$, nonstationary ESP holds.

Conversely, if nonstationary ESP holds, then Eq.~\eqref{eqn:mrc_converge} must converge, and there must be no input sequences such that $x_{t+1} - x_t  = 0$. Therefore, $0 < |a| < 1$ and $|b|$ are required, which proves the lemma.

\end{proof}

We can observe that $\frac{b}{au_t - 1}$ which appears in the proof above, corresponds to $\mathcal{G}(\mathbf{u}_t)^{-1}\mathbf{b}$. This is another resemblance between QRC and mRC.
\subsection{Proof of Cor.~\ref{cor:mrc-bounded}}
\begin{proof}
	Please note that if $x_0 = 0$, then $|x_1| = |b|$. Therefore, we can always assume $|x_0| > 0$ if $|b| > 0$. If $|a| \geq 1$ and $|b| > 0$, then for a constant input sequence $\{u_t = 1\}$,
\begin{equation}
\begin{aligned}
	x_{t} &= ax_{t-1} + b\\
	&\geq a^{t}x_0 + tb \underset{t \to \infty} \to \pm \infty.
\end{aligned}
\end{equation}
However, this contradicts the assumption of $x_t \in \mathcal{S}$ being bounded. Therefore, $|a| < 1$ is required.
From Thm.~\ref{thm:mrc-ns-esp}, it then satisfies nonstationary ESP. 

\end{proof}
\subsection{Preparation for proof of Rem.~\ref{rem:mrc-ipc}}
\begin{lem}{Variance of uniform distribution \cite{CaseBerg_2009}}
\begin{equation}
	\mathbb{E}_{u \sim \mathrm{Uniform}([-1, 1])}[u^2] = \frac{1}{3}.
\end{equation}	
\end{lem}
\begin{proof}
	The probability density function $\rho(x)$ of $\mathrm{Uniform}([-1, 1])$ in $u \in [-1, 1]$ is 
	\begin{equation}
		\rho(u) = \frac{1}{2}
	\end{equation}
Therefore, 
\begin{equation}
\begin{aligned}
		\mathbb{E}_{u \sim \mathrm{Uniform}([-1, 1])}[u^2] &=
		\int_{-1}^1 u^2\rho(u)du\\
		&= \int_{-1}^1 \frac{u^2}{2}du\\
		&= \left[\frac{u^3}{6}\right]_{-1}^1\\
		&= \frac{1}{3}.
	\end{aligned}
\end{equation}
\end{proof}

\subsection{Proof of Rem.~\ref{rem:mrc-ipc}}

\begin{proof}
Because $u_t \sim \mathrm{Uniform}([-1, 1])$, 	$\mathbb{E}_t[u_t] = 0$ and $\mathbb{E}_t[u_t^2] = \frac{1}{3}$. Of course, $u_i$ and $u_j$ are independent for all $i\neq j$. Therefore, $\mathbb{E}\left[(\sum_t u_t)^2\right] = \sum_t \mathbb{E}[u_t^2]$. In the following, please note that the time index $t$ of $x_t$ and $u_t$ has an offset of 1 because $x_{t+1}$ is the state after input of $u_t$ in our definition.

The covariance between $x_{t+1}$ and $u_{t-k}$ for $k \geq 0$ can be written as follows:
\begin{equation}
	\begin{aligned}
	\mathrm{Cov}[x_{t+1}, u_{t-k}] &= \mathbb{E}\left[(x_{t+1} - \mathbb{E}[x_{t+1}])(u_{t-k} - \mathbb{E}[u_{t-k}])\right]\\
	&= \mathbb{E}\left[(x_{t+1} - \mathbb{E}[x_{t+1}])u_{t-k}\right]\\
	&= \mathbb{E}[x_{t+1} u_{t-k}] - \mathbb{E}\left[\mathbb{E}[x_{t+1}]u_{t-k}\right]\\
	&= \mathbb{E}[x_{t+1} u_{t-k}] - \mathbb{E}[x_{t+1}]\mathbb{E}[u_{t-k}]\\
	&= \mathbb{E}[x_{t+1} u_{t-k}]\\
	&= a\mathbb{E}[x_t u_t u_{t-k}] + b\mathbb{E}[u_{t-k}]\\
	&= a\mathbb{E}[x_t u_t u_{t-k}]\\
	&= \frac{a\mathbb{E}[x_t]}{3}\delta_{k, 0}
	\end{aligned}
\end{equation}
The last transform above comes from the independence of $x_t$, $u_t$, and $u_{t-k}$ for $k\geq 1$.
In addition, 
\begin{equation}
\begin{aligned}
	\mathbb{E}[x_{t+1}^2] &= a^2\mathbb{E}[x_t^2u_t^2] + 2ab\mathbb{E}[x_tu_t]\\
	&=\frac{a^2}{3}\mathbb{E}[x_t^2].
\end{aligned}
\end{equation}
Therefore, 
\begin{equation}
\begin{aligned}
	C_{tot}^{MC} = C_{0}^{MC} &= \lim_{t\to\infty}\frac{\mathrm{Cov}^2[x_{t+1}, u_t]}{\mathrm{Var}[x_{t+1}]\mathrm{Var}[u_t]}\\
	&= \lim_{t\to\infty}\frac{a^2\mathbb{E}^2[x_t]}{3^2}\cdot \frac{3^2}{a^2\mathbb{E}[x_t^2]}\\
	&= \lim_{t\to\infty}\frac{\mathbb{E}^2[x_t]}{\mathbb{E}[x_t^2]}.
\end{aligned}
\end{equation}

If nonstationary ESP holds, then $|a| < 1$ and $|b| > 0$. In this case, we have
\begin{equation}
	\mathbb{E}[x_t] = b
\end{equation}
and
\begin{equation}
\begin{aligned}
	\mathbb{E}[x_t^2] &\underset{t \to \infty} \to b^2\left( a^2\mathbb{E}[u_{t}^2] + a^4\mathbb{E}[u_{t}^2] \mathbb{E}[u_{t-1}^2] + \cdots \right)\\
	&= \frac{b^2}{1 - \frac{a^2}{3}}.
\end{aligned}
\end{equation}
Therefore,
\begin{equation}
\begin{aligned}
	C_{tot}^{MC} = C_{0}^{MC} &= b^2\frac{1 - \frac{a^2}{3}}{b^2}\\
	&= 1 - \frac{a^2}{3}.
	\end{aligned}
\end{equation}
\end{proof}
\end{appendix}
\normalem
\bibliography{main.bib}
\end{document}